\documentclass{article}
\usepackage[utf8]{inputenc}
\usepackage{xcolor}
\usepackage{amsmath}
\usepackage{amssymb}
\usepackage{algorithm}
\usepackage{algorithmic}
\usepackage{mathtools}
\usepackage{cite}
\usepackage{amsmath,amssymb,amsfonts}
\usepackage{algorithmic}
\usepackage{graphicx}
  \graphicspath{{fig/}}
\usepackage{textcomp}
\usepackage{tikz}
\usepackage{csvsimple}
\usepackage{tkz-tab}
\usepackage{latexsym}
\usetikzlibrary{arrows}
\usetikzlibrary{snakes}
\usetikzlibrary{plotmarks} 
\usetikzlibrary{patterns}
\usetikzlibrary{decorations.markings}
\usepackage{pgfplots}
\usepackage{booktabs}
\usepackage{longtable}
\usepackage{authblk}
\usepackage[hyphens,spaces]{url}

\usepackage{subcaption}

\pgfplotsset{compat=1.14}

\newtheorem{theorem}{Theorem}
\newcommand{\qedsymbol}{\hspace{\fill}\rule{1.5ex}{1.5ex}}

\newtheorem{proposition}[theorem]{Proposition}
\newenvironment{proof}[1][Proof]{\noindent \textbf{#1.} }{\qedsymbol}

\def\Trials{[1\ldots n_k]}
\def\Iterations{[1\ldots n_r]}

\def\Steps{T}
\def\sumk{\sum_{k=1}^{n_k}}
\def\sumr{\sum_{r=1}^{n_r}}
\def\sumt{\sum_{t\in\Steps}}
\def\lagrangian{L(w,b,\xi,\eta,\lambda,\rho,\mu,\theta)}
\newcommand{\R}{\mathbb{R}}

\def\foralli{\forall i\in TS}
\def\forallinotrg{\forall i\in TS: y_i\neq 1}
\def\sumi{\sum_{i\in TS}}
\def\suminotrg{\sum_{i\in TS}^{y_i\neq 1}}
\def\sumip{\sum_{i^{\prime}\in TS}}
\def\sumipnotrg{\sum_{i^{\prime}\in TS}^{y_{i^{\prime}} \neq 1}}

\begin{document}

\title{Improving P300 Speller performance by means of optimization and machine learning}
\date{}
\author[1]{Luigi Bianchi\thanks{luigi.bianchi@uniroma2.it}}
\author[1]{Chiara Liti\thanks{chiara.liti@uniroma2.it}}
\author[2]{Giampaolo Liuzzi\thanks{giampaolo.liuzzi@diag.uniroma1.it}}
\author[1]{Veronica Piccialli\thanks{veronica.piccialli@uniroma2.it}}
\author[1]{Cecilia Salvatore\thanks{cecilia.salvatore@alumni.uniroma2.eu}}
\affil[1]{dept. of Civil Engineering and Computer Science, \textit{University of Rome Tor Vergata}}
\affil[2]{Department of Computer, Control and Management Engineering, \textit{Sapienza University of Rome}}
\maketitle

\section*{Abstract}
Brain-Computer Interfaces (BCIs) are systems allowing people to interact with the environment bypassing the natural neuromuscular and hormonal outputs of the peripheral nervous system (PNS). These interfaces record a user’s brain activity and translate it into control commands for external devices, thus providing the PNS with additional artificial outputs. 
In this framework,  the BCIs based on the P300 Event-Related Potentials (ERP), which represent the electrical responses recorded from the brain after specific events or stimuli, have proven to be particularly successful and robust. The presence or the absence of a  P300  evoked potential within the EEG features is determined through a classification algorithm. Linear classifiers such as stepwise linear discriminant analysis (SWLDA) and support vector machine (SVM) are the most used discriminant algorithms for ERPs’ classification. Due to the low signal-to-noise ratio of the EEG signals, multiple stimulation sequences (a.k.a. iterations) are carried out and then averaged before the signals being classified.
However, while  augmenting the number of iterations improves the Signal-to-Noise Ratio (SNR), it also slows down the process. In the early studies, the number of iterations was fixed (no stopping), but recently, several early stopping strategies have been proposed in the literature to dynamically interrupt the stimulation sequence when a certain criterion is met to enhance the communication rate. In this work, we explore how to improve the classification performances in P300 based BCIs by combining optimization and machine learning. First, we propose a new decision function that aims at improving classification performances in terms of accuracy and Information Transfer Rate both in a no stopping and early stopping environment. Then, we propose a new SVM training problem that aims to facilitate the target-detection process. Our approach proves to be effective on several publicly available datasets.

\section{Introduction}\label{sec:Intro}

A Brain-Computer Interface (BCI) is a system that records a user's brain activity and allows him to interact with the environment by exploiting both signal processing and machine learning algorithms. In most cases, the recorded signals are noisy, so that filtering or averaging techniques are used to improve the signal-to-noise ratio (SNR). The information embedded in signals that are relevant to characterize the user's mental states are then selected during a feature extraction procedure before being classified and translated into artificial outputs \textemdash i.e. into control commands for an output device such as a pointer, a keyboard or a robotic arm \cite{lotte2014tutorial,lotte2018review,QuitadamoUML, wolpaw2012brain, mccane2015p300}. BCIs use either electrical, magnetic and metabolic signals \cite{wolpaw2012brain} recorded with methods such as electroencephalography (EEG), electrocorticography (ECoG), magnetoencephalography (MEG), functional Near Infra-Red Spectroscopy (fNIRS) and functional Magnetic Resonance Imaging (fMRI).\\
In this framework, BCIs based on event-related potentials (ERPs) have proven to be particularly successful and robust \cite{schreuder2013optimizing}. ERPs represent the electrical responses recorded from the brain through EEG techniques after specific events or stimuli. The ERPs are embedded within the general EEG activity \cite{sur2009event}, and are time-locked to the processing of a specific stimulus. As their amplitude is lower that the one of the ongoing EEG activity, averaging techniques are employed to increase the SNR: in principle, averaging background noise which is not correlated to an event, such as the ongoing EEG activity, tends to reduce its contribution to a small offset, which can be easily filtered out, while the evoked responses, supposed to be the same after each stimulus, are left unmodified.  
An ERP-based BCI attempts to detect ERP components to infer the stimulus that the user intended to choose \textemdash i.e. the stimulus eliciting the ERP components \cite{treder2010c}.\\
In 1988, the P300 ERP was first used by Farwell and Donchin within a BCI system \cite{farwell1988talking}. Their P300 Speller consists of 36 alpha-numeric characters arranged within the rows and columns of a $6 \times6$ matrix. The user's task is to focus the attention on a specific character \textemdash i.e. on one of the cells of the matrix. Each of the 6 rows and 6 columns then flashes for few tenths of milliseconds in a random sequence. A sequence of 12 different flashes \textemdash the 6 rows and 6 columns \textemdash is called an iteration. It constitutes the basis of an \textit{oddball paradigm} in which two classes of stimuli, namely the target (or rare) and the non target (or frequent), which occur with different probabilities (0.166 and 0.833 in this case), and that elicit two different brain responses. In particular, the target (rare) stimuli should elicit the P300 response which is not evoked after non target (frequent) stimuli. In our case the row and the column containing the attended character represent the target stimuli, while the other ten are the non-target ones. Brain responses to the target and non-target stimuli are distinguished using a classification algorithm. The correct identification of the target row and column allows the desired character's selection, which is located at their intersection \cite{krusienski2006comparison,krusienski2008toward, sellers2006p300}.\par
Later on in the literature, different variations of the original P300 paradigm have been developed in order to improve the speller framework. For instance, in \cite{schaeff2012exploring, schreuder2011listen, treder2011gaze} the authors proposed gaze-independent spellers, i.e. communication systems that can be used by subjects who have impairment at moving their eyes. In all speller paradigms, given a sentence/run to copy-spell, the EEG data are organized in terms of trials, iterations, and sub-trials. A single character selection step is here referred to as a trial. Each trial consists of several iterations/stimulation
sequences, during which all the stimuli are intensified once in a pseudo-random order. 
A single stimulus intensification is here referred to as a sub-trial. The trials' selection process can involve one ore two levels. In the former case, symbols are typically presented successively thus involving a single selection step. In the latter, the user has to select a group of symbols first and then the target symbol.\par
To use a BCI, two phases namely training/calibration and test/online are typically required. During the calibration phase, the user focuses his/her attention on a specific character. The acquired EEG signals are then preprocessed by filtering. A subset of EEG features is extracted to represent the signal in a compact form. The obtained EEG patterns are recognized using a classification algorithm, which is trained on the subset of identified features to determine the presence or the absence of a P300 evoked potential. In the online phase, new EEG patterns are classified using the trained model before being translated into a command for an application.
As described above, in ERP-based BCIs, to perform a single selection step, multiple iterations are carried out to improve the SNR. However, repeated stimulations increase the time necessary
to detect the brain signals reducing the communication rate. In this work, we explore how to improve the classification performance by combining optimization and machine learning.  

\subsection{Literature Review}\label{subsec:literature}
As mentioned above, the presence or the absence of a P300 evoked potential within the EEG features is determined using a classification algorithm \cite{krusienski2006comparison}.\\
Formally, the detection of brain responses to the target and non-target stimuli can be translated into a binary classification problem.
Let $TS$ be the training set defined as:
\begin{multline}\label{training_set_compact}
    TS = \{(x_i, y_i): x_i\in \R^n, y_i\in S\subset\{-1,+1\},\\ i\coloneqq(k,r,t,f)\quad \forall k=1\dots n_k,\  r=1\dots n_r,\ t\in\Steps,\ f =1\dots n_f \}
\end{multline}
where $n_k$ denotes the total number of trials in the training phase and $n_r$ denotes the number of iterations for each trial; the number of flashes $n_f$ and the set of levels $\Steps$ together denote the set of possible stimuli that compose the stimulation sequence (i.e. $n_f = 6$ and $\Steps=\{row,\ column\}$ for P300 Speller's paradigm or $\Steps=\{outer,\ inner\}$ for two-levels paradigms).

During the calibration phase, a classification algorithm is trained over $TS$ to learn the discriminant function $f$ such that
\begin{equation}\label{eq:mapping_function}
f(x)=y,
\end{equation}
and this function is used in the online phase to spell words or sentences.
In the BCI literature, several algorithms have been proposed for addressing this classification problem \cite{lotte2018review}. In particular, linear classifiers such as stepwise linear discriminant analysis (SWLDA) \cite{draper1998applied}, and support vector machine (SVM) \cite{friedman2001elements} are still the most used discriminant algorithms for ERPs' classification \cite{lotte2018review}.
These methods classify the brain responses by means of a separating hyperplane \cite{krusienski2006comparison}. This discriminant function is built on the basis of the training data, and it is defined as:
\begin{equation}\label{hyperplane}
  f(x)= w^Tx + b,
\end{equation}
where $w$ is the vector containing the classification weights and $b$ is the bias term. Linear classifiers differ in the way they learn $w$ and $b$ \cite{krusienski2006comparison}.  In \eqref{hyperplane}, the right-hand side is called decision value. Its absolute value is proportional to the distance of the sample points $x$ from the separating hyperplane.\par
In a standard binary classification problem, for each instance the class label is assigned based on the sign of the relative decision value. However, in a classical P300 Speller \cite{farwell1988talking}, based on the assumption that a P300 is elicited for one of the six row/column stimuli and finding that the P300 response is invariant to row/column stimulation, the target class is assigned to the stimuli matching the maximum decision values for both the rows and the columns  \cite{krusienski2006comparison}. In general, remembering the definition of $\Steps$ and $n_f$ given in \eqref{training_set_compact}, we can identify the target stimulus for trial $k\in\Trials$ and iteration $r\in\Iterations$ as: 
\begin{equation}
	\begin{split}
		\mbox{predicted stimulus}_{(k,r,t)} & =\underset{f=1\dots n_f}{\mathrm{argmax}}\left[w^{T}x_{(k,r,t,f)} +b\right] \quad \forall t\in\Steps \\
	\end{split}
\end{equation}
The predicted character for trial $k\in\Trials$ and iteration $r\in\Iterations$ is then identified by combining the predicted target stimuli found $\forall t\in\Steps$ (i.e a row target and a column target for the standard P300 paradigm).\par

As mentioned in Section \ref{sec:Intro}, for each character, data recorded from multiple iterations have to be integrated to improve the SNR. To the best of our knowledge there exist two main different iteration-averaging strategies in the literature: (i) \textbf{ERP avg}: for each character brains responses to target and non target stimuli are averaged across the iterations before being classified, and (ii) \textbf{DV med}: for each character the decision values of each target and non-target stimulus are averaged across the iterations before assigning the target class. Recently in \cite{noiES:2019}, a new classification function namely score-based function (SBF) has been introduced for integrating brain responses recorded from multiple iterations. For each character, the SBF exploits a set of heuristically-determined scores to weight each stimulus according to its decision value. For each stimulus, the assigned scores are summed up iteration by iteration. The target class (one for the row and one for the column) is assigned to the stimulus having the highest total score at the last available iteration. The SBF has been introduced for developing an early stopping method (ESM) \textemdash i.e. an automatic method that interrupts the stimulation at any point in a trial when a certain criterion, based on the ongoing classification results, is satisfied (see for instance \cite{Lenhardt,Zhang,Liu,Hone,schreuder2011listen,Jin,Throckmorton2013,Mainsah2014healthy,Jiang2018,Vo2017,Vo:2018paper, schreuder2013optimizing}). The proposed ESM outperformed the current state-of-the-art early stopping methods. 

In this paper, we follow the same line of research of \cite{noiES:2019}, by making some further steps to include the information of the protocol into the classification phase. Indeed, the novelty of our approach consists of three points:
\begin{enumerate}
     \item determine the optimal scores for each participant by solving an optimization problem on her/his training data; 
    \item solve a modified version of the optimization problem in order to implement an efficient early stopping method;    \item include the information on the decision function (the target is the stimulus having maximum decision value) into the training problem

\end{enumerate}
The great advantage of our method, is that the calibration phase (different for each participant) becomes completely automatic and does not need any cross validation phase or manual parameters tuning.

The paper is structured as follows: in Section \ref{sec:opt}, we introduce our new decision function, defining the optimization problems to be solved both in the no stopping and early stopping scenario. In Section \ref{sec:newtr}, we introduce a new training problem that keeps into account explicitly the target assignment in BCI, and in Section \ref{sec:wolfe} we derive its Wolfe dual. In Section \ref{numres} we report the behavior of our new approaches on several datasets and finally we draw some conclusions in Section \ref{sec:conc}.

\section{An optimized score based decision function}\label{sec:opt}
In \cite{noiES:2019}, a set of heuristically-determined scores has been used to weight and combine the decision values of multiple iterations within an early stopping setting. In this work, we decided to modify the approach by using a set of scores automatically determined by solving a mixed integer linear programming (MILP) problem for each participant. 
Each stimulus receives a weight according to its decision value: five zones are defined, and each zone gets a different score \emph{a,b,c,d,e}. In particular, the scores are related to the confidence in the classification of the given stimulus as target: the score $a$ is assigned to the stimulus that is most likely to be the target, whereas the stimuli that are highly unlikely to be the target get score $e$. All the stimuli in the middle get decreasing scores according to the distribution of the decision values.

The zones are identified by considering the decision values of all iterations for all stimuli in the training set and computing the corresponding quartiles Q1, Q2 and Q3. The idea is to produce scores that reflect the distribution of the data. 
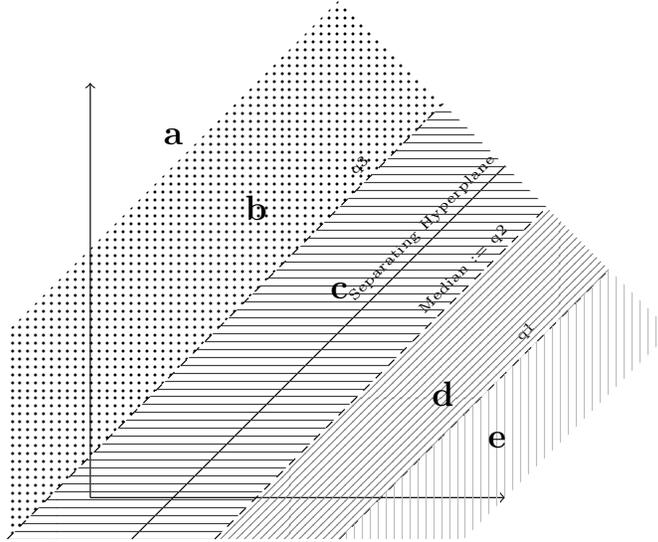
\begin{figure}
\centering
\resizebox{0.75\linewidth}{!}{
\begin{tikzpicture}
           \draw[->] (0,0)--(5,0) node[right]{}; 
           \draw[->] (0,0)--(0,5) node[above]{}; 

          \draw (0.5,-0.5) -- (5,4); 
          \draw[dashed] (-1,-0.5) -- (4.25,4.75); 
          \draw[dashed] (1.5,-0.5) -- (5.5,3.5) ; 
          \draw[dashed] (3,-0.5) -- (6.25,2.75); 

          \draw (4,3.25) node[rotate=45,font=\tiny]
              {Separating Hyperplane};

          \draw (4.5,2.75) node[rotate=45,font=\tiny]
              {Median := q2};

          \draw (3.25,4) node[rotate=45,font=\tiny]
              {q3};

          \draw (5.25,2) node[rotate=45,font=\tiny]
              {q1};

          \draw[black,pattern=dots, pattern color=black, draw=none] (-1,-0.5) -- (-1,2) -- (3,6) -- (4.25,4.75)-- cycle;

          \draw[darkgray,pattern=horizontal lines,pattern color=darkgray, draw=none] (1.5,-0.5) -- (-1,-0.5) -- (4.25,4.75) -- (5.5,3.5) -- cycle;

          \draw[gray,pattern=north east lines,pattern color=gray, draw=none] (1.5,-0.5) -- (5.5,3.5) -- (6.25,2.75)-- (3,-0.5) -- cycle;

          \draw[lightgray,pattern=vertical lines,pattern color=lightgray, draw=none] (3,-0.5) -- (6.25,2.75) -- (7,2) -- (4.5,-0.5) -- cycle;
         \draw (1,4.35) node[font=\large] {\textbf{a}};
           \draw (2,3.5) node[font=\large] {\textbf{b}};
           \draw (3,2.5) node[font=\large] {\textbf{c}};
           \draw (4.25,1.25) node[font=\large] {\textbf{d}};
           \draw (4.9,0.7) node[font=\large] {\textbf{e}};           

  \end{tikzpicture}
  }
\caption{Graphical representation of the score distribution, reflecting the displacement of the points
w.r.t the distribution of the decision values. The different areas represent the
confidence of the classification w.r.t the target class.}
\label{fig:score_rule}
\end{figure}

  







Figure \ref{fig:score_rule} shows how the scores are assigned depending on the distribution of the quartiles of the decision values.
The maximum score $a$ is assigned only if the confidence in the current classification is extremely high: i.e. if the decision value is positive and higher than all the other decision values of the current iteration.


Note that, given the separating hyperplane, the score assignment for each stimulus of each character is known: so, it is possible to build the following binary vectors that represent in a compact form the score vector assignment $z$ for each stimulus of each character:
\[z_s^{k,r,t,f} =\left\{\begin{array}{ll}
    1 & \mbox{if stimulus $f$ of level $t$ gets score $s$ at iteration $r$ for char $k$} \\
    0 & \mbox{otherwise}
\end{array}\right.
\]
where  $f=1\dots n_f$ and $t\in\Steps$ identify the stimulus, $k=1\dots n_k$ identifies the character, $r=1\dots n_r$ identifies the iteration and, finally, $s=a,\ldots, e$ identifies the score. The score assignments depends on the kind of the primary aim of the BCI: 
\begin{description}
\item[(i)]if the main focus is the accuracy, the idea is to use all the available iterations for spelling a character (no stopping protocol), also in the online phase.
\item[(ii)] if the idea is to try and speed up the communication, then the performance to be maximized is the transmission rate, trying to reduce the number of iterations needed to spell a character in the online phase (early stopping). 
\end{description}
In the next two subsections, we describe the  Mixed Integer Linear Programming (MILP) Problems we define in order to find the scores in the two different settings.
\subsection{No stopping OSBF}
First, we propose a strategy to choose the scores when all the iterations are exploited and the primary focus is to increase the classification accuracy. In this setting, we aim to reliability of the classification and we do so by imposing the following constraints:

\begin{enumerate}
    \item at the last iteration, we require, if possible, that the score obtained by the target stimulus is larger (with some margin if possible, that implies robustness of the classification) than the score of any non target stimulus. This means that we ask not to fail in the classification after the last available iteration; if this is not possible, a suitable binary variable representing the failure on that stimulus is set to one;
    \item to make the classification more robust on the test set, we require that in as many iterations as possible, the score of the target is larger than the one of the non target stimuli;
    \item as an objective, we try and maximize the accuracy on the training set, and the number of iteration where the classification is robust.
\end{enumerate}
Our main variable in the optimization problem is the vector of scores $s=\left(a~b~c~d~e\right)^T$.

We add an auxiliary variable to try and impose some distance between the score of the target stimulus and the scores of the non target stimuli that we call $\Delta$, and that represents a measure of reliability of the classification. Further, we add some binary variables:
\begin{itemize}
\item $x_t^{k,r}$: binary variable that is equal to $1$ if the target of character $k$ for level $t$ has a score at iteration $r$ that is larger than the score of any non target stimulus plus $\Delta$ 
\item $err_t^{k}$: binary variable that is equal to $1$ if the target is not correctly classified  for character $k$ at level $t$, i.e. if at the last iteration the target score is lower or equal to the score of some non target stimulus
\end{itemize}
The MILP problem to be solved is then the following:
\begin{align}\label{eq:wobj}
\max &\displaystyle \sumt\Bigg(\Big(1-\frac{\sumk err_t^k}{n_k}\Big) + \frac{1}{n_kn_r}\sumk\sumr x_t^{k,r}\Bigg)
 \\\label{eq:wub}
& s_1\le u\\\label{eq:word}
& s_{j+1}\le s_j-1\,\forall j=1,\ldots,4\\\label{eq:w3nneg}
&s_3\ge 0\\\label{eq:wlb}
&s_5\ge l\\\label{eq:nostop2it}
&\displaystyle \Delta\ge s_1-s_5+1\\\label{eq:fail}
&\displaystyle 1-err_t^k\le x_t^{k,nr}\,\forall k=1\dots n_k,\ \forall t\in\Steps\\\label{eq:ES}
& \displaystyle \sum_{r=1}^{\bar r}s^Tz^{k,r,t,f}+ \Delta \le \sum_{r=1}^{\bar r}s^Tz^{k,r,t,trg(k,t)}\\\nonumber&\displaystyle\quad+M(1-x_t^{k,\bar r})\quad\forall k,\bar r,t, f\colon f\not=trg(k,t)\\
& err_t^{k},x_{t}^{k,r} \mbox{ binary }\forall k~\forall t~\forall r
\end{align}

where $l$ and $u$ are chosen bounds on the possible values of the scores, and $M$ is large enough to make the constraints trivially satisfied when the corresponding binary variable $x_t^{k,r}$ is zero. 
The objective function, that has to be maximized, is composed by two terms: the percentage of success on the training set, and the average number of iterations where the classification is robust and reliable.
We then have the following constraints:
\begin{description}
\item[(i)] Constraints (\ref{eq:wub}), (\ref{eq:word}) and (\ref{eq:wlb}) impose that the scores are bounded and that are ordered in decreasing order and differ of at least one; whereas constraint (\ref{eq:w3nneg}) imposes that the first three scores are nonnegative
\item[(ii)] constraint (\ref{eq:nostop2it}) imposes a lower bound on the threshold to ensure reliability of the classification. Indeed this lower bound ensures that the threshold has a minimum value depending on the scores: in particular $s_1-s_5+1$ represents the maximum difference in score that can be assigned to different flashes in a single iteration. Therefore, even in the worst possible scenario, where two flashes get the same score, there must be at least one iteration where one gets the maximum score and the other the minimum score to break the parity.
\item[(iii)] constraints (\ref{eq:fail}) impose that  variable $err_t^k$ is 1 if and only if $x_t^{k,n_r}=0$, that is it represents an unreliable classification at the last iteration.
\item[(iv)] constraints (\ref{eq:ES}) impose that if at iteration $\bar r$ the classification is reliable for the target $trg(k,t)$ of character $k$ at level $t$, then the corresponding binary variable $x_t^{k,r}$ is set to 1.
\end{description}

\subsection{Early Stopping OSBF}

Problem (\ref{eq:wobj}) can be modified in order to  improve the system performance in terms of speed, implementing an automatic Early Stopping Method, similarly to \cite{noiES:2019}.

The idea is again to use the scores $s$ and the threshold $\Delta$ at each iteration of the test phase to verify an early stopping condition: during the test phase, the stimuli are ordered according to the sum of their scores and, if the difference in score between the first and second stimulus is greater than the threshold $\Delta$, the method classifies the target character and the remaining iterations are not performed. 

In order to adapt problem (\ref{eq:wobj}) to the early stopping setting, we introduce some further constraints, and modify the meaning of some binary variables:
\begin{align}\label{eq:wobjES} \displaystyle
\max &~ \sumt\Bigg( 1-\frac{1}{n_k}\sumk err^k_t -\frac{100\times n_{fl}}{60}\frac{SOA}{n_k}\left(\sumk\sumr r x_t^{k,r}+n_r\sumk{err_t^k}\right)\Bigg)
 \\\nonumber &(\ref{eq:wub})-(\ref{eq:nostop2it})\\
\label{eq:failES}
&\displaystyle 1-err_{t}^k\le \sumr x_{t}^{k,r}\,\forall k\in \Trials\\
\label{eq:oneitES}&\displaystyle \sumr x_{t}^{k,r}\le 1\,\forall k\in\Trials\\
\nonumber &(\ref{eq:ES})\\
\label{ES1}& \displaystyle \sum_{r=1}^{\tilde r}s^Tz^{k,r,t,f}\ge \sum_{r=1}^{\tilde r}s^Tz^{k,r,t,f}-\Delta+1\\\nonumber
&\displaystyle\quad-M(1-x_{t}^{k,\bar r})\quad\forall k, \forall t,\forall \tilde r<\bar r, \forall f\colon f\not=trg(k,t)\\
& err_{t}^{k},x_{t}^{k,r} \mbox{ binary }\forall k~\forall t~\forall r
\end{align}

In this case, the objective function keeps into account both the percentage of success (to be maximized) and the time needed for classification (to be minimized). Note that the second term (which represents the trial duration in minutes) was multiplied by a factor $100$ to make the two terms of the objective function comparable. 
We then have some further constraints, since in this case we are interested in the first iteration where the following early stopping condition is met:
\begin{equation}\label{EScond}
\displaystyle \sum_{r=1}^{\bar r}s^Tz^{k,r,t,f}+ \Delta \le \sum_{r=1}^{\bar r}s^Tz^{k,r,t,trg(k,t)}.
    \end{equation}
In this model, we set the binary variables $x_t^{k,r}$ in such a way that it is 1 if and only if the early stopping condition (\ref{EScond}) is verified for the first time on the target at iteration $r$, and it is not satisfied by any non target stimulus earlier. This is imposed by the combination of  constraints (\ref{eq:ES}),  (\ref{eq:oneitES}) and (\ref{ES1}).
\par We stress that in both the no stopping and the early stopping scenarios, the MILP problem is solved using the training set data (the same used to build the hyperplane), whereas the score efficiency is evaluated on the test set.

\section{A new training problem}\label{sec:newtr}

As already pointed out in the introduction, in order to achieve a good classification accuracy it is fundamental to exploit the information that at each iteration there is exactly one target stimulus for each level, assigning then the target class to the stimulus having the maximum decision value. Our idea is to try and add this protocol knowledge already in the training problem.

Given the definition of training set given in \ref{training_set_compact}, the standard  training problem to solve in order to find a separating hyperplane according to the SVM approach is the following \cite{surveySVM}:

\begin{align}
    \min_{w\in\R^n, b\in\R} \frac{1}{2}\|w\|^2 + C_1\sum_{i\in TS}\xi_i
\end{align}
\begin{align}
    &y_i(w^Tx_i+b)\ge 1-\xi_i\qquad &\foralli\\
    &\xi_i\ge 0 \qquad &\foralli
\end{align}

In this work, we modify the training problem including the information that the target stimuli should receive the maximum decision value among all the other flashes. Let's denote by $trg_i$ the target stimulus for the stimulation sequence where the stimulus $i$ belongs: so, in particular, if $i = (k,r,t,f)$ we will have: \[trg_i = (k,r,t,f^{\prime})\in TS \quad\&\quad y_{trg_i} = 1\]

Then, we want to impose:
\begin{equation}\label{eq:maxcext}
\begin{split}
    w^Tx_{(k,r,t,f_1)} +b\ge w^Tx_{(k,r,t,f_2)}+b \qquad  \forall& k,r,t,f_1,f_2\colon\\&y_{(k,r,t,f_1)} = 1\ \&\ y_{(k,r,t,f_2)} = -1 
\end{split}
\end{equation}
From now on, in order to simplify the notation, we will write constraints \eqref{eq:maxcext} in the following more compact form:
\begin{equation}
    w^Tx_{trg_i} +b\ge w^Tx_i+b \qquad  \forallinotrg
\end{equation}
and we add slack variables to avoid infeasibility, getting the following set of constraints:
\begin{align}\label{eq:maxconstr}
    &w^Tx_{trg_i} - w^Tx_i \ge 1 -\eta_i &\qquad  \forallinotrg \\  
    &\eta_i\ge 0 &\qquad  \forallinotrg
\end{align}

Now we simply plug these constraints into the primal problem getting the new training problem based on the maximum decision function:

\begin{align}
    &\min_{w\in\R^n, b\in\R} \frac{1}{2}\|w\|^2 + C_1\sum_{i\in TS}\xi_i +C_2\suminotrg\eta_i\label{fo:maxSVM}    
\end{align}
\begin{align}
    &y_i(w^Tx_i+b)\ge 1-\xi_i\qquad &\foralli\label{vinc:sign}\\
    &w^Tz_i \ge 1 -\eta_i &\qquad  \forallinotrg \label{vinc:max}\\ 
    &\xi_i\ge 0 \qquad &\foralli \label{vinc:nonneg_sign}\\
    &\eta_i\ge 0 &\qquad  \forallinotrg\label{vinc:nonneg_max}
\end{align}

where the vector $z$ is defined as:
\[z_i = x_{trg_i} - x_i \quad \forallinotrg\]

\section{Wolfe Dual of the new training problem}\label{sec:wolfe}
In order to build the Wolfe Dual of the quadratic optimization problem~(\ref{fo:maxSVM})-(\ref{vinc:nonneg_max}), it is necessary to introduce the dual multipliers of the constraints:
\begin{itemize}
    \item $\lambda_i\quad\foralli$: the multiplier associated to constraints~\eqref{vinc:sign}
    \item $\rho_i\quad \forallinotrg$: the multiplier associated to constraints~\eqref{vinc:max}
    \item $\mu_i\quad\foralli$: the multiplier associated to constraints~\eqref{vinc:nonneg_sign}
    \item $\theta_i\quad \forallinotrg$: the multiplier associated to constraints~\eqref{vinc:nonneg_max}
\end{itemize}

Let us define the vector $\lambda$ and $\rho$ as the vectors of size $l_1$ and $l_2$ respectively containing $\lambda_i \ (\foralli)$ and $\rho_i \ (\forallinotrg)$. Then we define the following matrix $\Sigma\in\Re^{(l1+l2)\times n}$:
\[
\Sigma = \begin{pmatrix}y^1(x^1)^T \cr\vdots \cr y^{l_1}(x^{l_1})^T\cr (z^1)^T\cr\vdots\cr (z^{l_2})^T \end{pmatrix}
\]
The following proposition holds:
\begin{proposition}\label{wolfedual}
The dual problem of problem (\ref{fo:maxSVM}) is
\begin{align}
\min & \frac 1 2 \begin{pmatrix}\lambda^T &\rho^T \end{pmatrix}\Sigma\Sigma^T\begin{pmatrix}\lambda \cr\rho\end{pmatrix}-e^T\lambda-e^T\rho \\
& y^T\lambda=0\\ 
&0\le  \lambda\le C_1e\\
&0\le  \rho\le C_2e
\end{align}
\end{proposition}
\begin{proof}
The Wolfe dual of problem~\eqref{fo:maxSVM}-\eqref{vinc:nonneg_max} is given by:

\begin{align}\label{eq:newdualprob}
    \max_{w,b,\lambda,\rho,\mu,\theta}&\lagrangian\\
    &\nabla_{w}\lagrangian = 0\label{eq:w_wolfe}\\
    &\nabla_{b}\lagrangian = 0\label{eq:b_wolfe}\\
    &\nabla_{\xi}\lagrangian = 0\label{eq:xi_wolfe}\\
    &\nabla_{\eta}\lagrangian = 0\label{eq:eta_wolfe}\\
    &\lambda,\rho,\mu,\theta \ge 0\label{eq:nonneg_wolfe}
\end{align}
where $\lagrangian$ is the Lagrangian of optimization problem~(\ref{fo:maxSVM})-(\ref{vinc:nonneg_max}) that can be expressed as follows:

\begin{equation}\label{eq:lagrangian1}
    \begin{split}
        \lagrangian &= \frac{1}{2}\|w\|^2 + C_1\sumi\xi_i +C_2\suminotrg\eta_i-\sumi\lambda_i\left(y_i(w^Tx_i+b)-1+\xi_i\right)+\\
        &-\suminotrg\rho_i\left(w^Tz_i-1+\eta_i\right)-\sumi\mu_i\xi_i-\suminotrg\theta_i\eta_i
    \end{split}
\end{equation}
By rearranging terms equation~\ref{eq:lagrangian1} can be rewritten as:
\begin{equation}\label{eq:lagrangian}
\begin{split}
    &\lagrangian =\frac{1}{2}\|w\|^2 + \sumi\xi_i(C_1-\lambda_i-\mu_i)+\sumi\lambda_i+\suminotrg\rho_i+\\
    &+\suminotrg\eta_i(C_2-\rho_i-\theta_i)-w^T\left(\sumi\lambda_iy_ix_i+\suminotrg\rho_iz_i\right)-b\sumi\lambda_iy_i
\end{split}
\end{equation}
The constraints of the Wolfe Dual (equations~\ref{eq:w_wolfe}-\ref{eq:eta_wolfe}) can now be computed based on the Lagrangian function in equation~\ref{eq:lagrangian}.
The equation $\nabla_w\lagrangian = 0$ leads to an expression for $w$:
\begin{equation}\label{eq:w}
    w=\left(\sumi\lambda_iy_ix_i+\suminotrg\rho_iz_i\right),
\end{equation} whereas the equation $\nabla_b\lagrangian = 0$
leads to the constraint
\begin{equation}\label{eq:b}
    \sum_{i\in TS}\lambda_iy_i = 0
\end{equation}
Equation    $\frac{\partial \lagrangian}{\partial\xi_i} = 0 $ allows to derive $\mu_i$ as a function of $\lambda$: 
\begin{equation}\label{eq:xi}
    C_1-\lambda_i-\mu_i = 0\qquad\foralli
\end{equation}
whereas $\frac{\partial \lagrangian}{\partial\eta_i} = 0$
results in an expression of $\theta_i$ as a function of $\rho_i$
\begin{equation}\label{eq:eta}
    C_2-\rho_i-\theta_i = 0\qquad\forallinotrg
\end{equation}
Non-negativity of the multipliers $\lambda,\ \rho,\ \mu,\ \theta$ 
combined with equations~\eqref{eq:xi} and~\eqref{eq:eta} result in the following set of constraints:
\begin{align}
    &0\le\lambda_i\le C_1\quad&\foralli\label{eq:boxlambda}\\
    &0\le\rho_i\le C_2\quad&\forallinotrg\label{eq:boxrho}
\end{align}
We can plug  equations~\eqref{eq:xi} and~\eqref{eq:eta} in the objective function, getting:
\begin{equation}
\begin{split}
\label{eq:lagrangian2}
&\lagrangian =\frac{1}{2}\|w\|^2 + \sumi\xi_i(C_1-\lambda_i-\mu_i)+\sumi\lambda_i+\suminotrg\rho_i+\\
&+\suminotrg\eta_i(C_2-\rho_i-\theta_i)-w^T\left(\sumi\lambda_iy_ix_i+\suminotrg\rho_iz_i\right)-b\sumi\lambda_iy_i = \\
    &= \frac{1}{2}\|w\|^2 + \sumi 0\times\xi_i+\sumi\lambda_i+\suminotrg\rho_i+\suminotrg 0\times\eta_i-w^Tw-0\times b =\\
    &= -\frac{1}{2}\|w\|^2+\sumi\lambda_i+\suminotrg\rho_i
\end{split}
\end{equation}
The Wolfe Dual of problem~\eqref{fo:maxSVM}-\eqref{vinc:nonneg_max} can then be expressed by using equation~\eqref{eq:lagrangian2} as objective and equations~\eqref{eq:w}, \eqref{eq:b}, \eqref{eq:boxlambda}, \eqref{eq:boxrho} as constraints.

\begin{align}
    &\min \frac{1}{2}\|w\|^2-\sumi\lambda_i-\suminotrg\rho_i\label{eq:fo-wolfe1}\\
    &w=\left(\sumi\lambda_iy_ix_i+\suminotrg\rho_iz_i\right)\\ 
    &\sumi\lambda_iy_i = 0\\
    &0\le\lambda_i\le C_1\qquad\foralli\\
    &0\le\rho_i\le C_2\qquad\forallinotrg
\end{align}

Note that:
\begin{equation*}
\begin{split}
    &\|w\|^2 = w^Tw=\left(\sumi\lambda_iy_ix_i+\suminotrg\rho_iz_i\right)^T\left(\sumip\lambda_{i^{\prime}}y_{i^{\prime}}x_{i^{\prime}}+\sumipnotrg\rho_{i^{\prime}}z_{i^{\prime}}\right)=\\
    &\sumi\sumip (\lambda_i\lambda_{i^\prime}y_iy_{i^\prime}(x_i)^Tx_{i^\prime})+\suminotrg\sumipnotrg(\rho_i\rho_{i^\prime}(z_i)^Tx_{i^\prime})+2\sumi\sumipnotrg(\lambda_i\rho_{i^\prime}y_i(x_i)^Tz_{i^\prime})
\end{split}
\end{equation*}

Let us define the vector $\lambda$ and $\rho$ as the vectors of size $l_1$ and $l_2$ respectively containing $\lambda_i \ (\foralli)$ and $\rho_i \ (\forallinotrg)$. Then we define the following matrix $\Sigma\in\Re^{(l1+l2)\times n}$:
\[
\Sigma = \begin{pmatrix}y^1(x^1)^T \cr\vdots \cr y^{l_1}(x^{l_1})^T\cr (z^1)^T\cr\vdots\cr (z^{l_2})^T \end{pmatrix}
\]
The dual problem can then be rewritten as 
\begin{align}
\min & \frac 1 2 \begin{pmatrix}\lambda^T &\rho^T \end{pmatrix}\Sigma\Sigma^T\begin{pmatrix}\lambda \cr\rho\end{pmatrix}-e^T\lambda-e^T\rho \\
& y^T\lambda=0\\ 
&0\le  \lambda\le C_1e\\
&0\le  \rho\le C_2e
\end{align}
that is still a quadratic convex programming problem.
\end{proof}

\section{Numerical Results}\label{numres}
\subsection{Dataset}
We tested our approaches on five different datasets:
\begin{description}
\item[AMUSE] The protocol is based on auditory stimulus elicited by means of spatially located speakers, we have two levels, 15 rounds, six classes for each level, see Fig. \ref{fig:amuse} \cite{schreuder2011listen}. It is performed on healthy subjects and downloadable by the BNCI horizon website \cite{bnci}.
\begin{figure}
\includegraphics[scale=0.3]{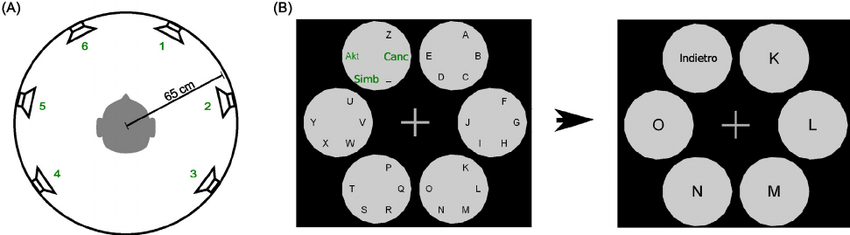}
\caption{Graphical representation of the AMUSE paradigm in which six speakers are places all over the subject. In the first level, each speaker is used to represent a set of characters, while on the second level each speaker is used to represent a single character among the previously selected set.}
\label{fig:amuse}
\end{figure}
\item[P300 Speller] The protocol is the classical P300 Speller \cite{farwell1988talking}, performed on 10 healthy subjects. 
\item[ALS P300 Speller] The protocol is the classical P300 Speller \cite{farwell1988talking}, performed on 8 patients suffering of Amyotrophic Lateral Sclerosis (ALS).
\item[MVEP] It is a visual protocol in which a moving pattern generates a movement-onset visual evoked potential that is used to recognize the user’s choice. This protocol is based on modifications of Cake Speller protocol\cite{treder2011gaze}. Sixteen healthy subjects have been involved in the study.
\item[Center Speller] It is a visual protocol where we have a visual stimulus elicited by means of three different  stimuli, two levels, 10 rounds, six classes for each level  \cite{treder2011gaze}. It is performed on 13 healthy subjects.
\item[Akimpech] It is a P300 Speller performed on 27 healthy subjects, the number of characters is 16 with 15 iterations for each character in the calibration phase, whereas in the online phase changes depending on the subject.
\end{description}
 All EEG signals were pre-processed and features were extracted with the NPXLab Suite \cite{bianchi2018}.
 Details of the datasets are reported in Table~\ref{table:dataset parameters}. Please note that we have evaluated our strategy on EEG data recorded from 95 subjects thus assessing its generalization capabilities.  
 \par
 Two principal pre-processing operations were applied:
 \begin{itemize}
     \item Electrodes selection: for the datasets Center Speller, MVEP, and AMUSE we kept the electrodes belonging to the 10-20 EEG placement. This strategy allows us to reduce both the dimension of the dataset and the overfitting;
     \item k-decimation: this technique was applied to all datasets in order to reduce overfitting. In this case, we down-sampled the EEG signal from every electrode by replacing each $k$ consecutive samples with their average value.
 \end{itemize}

 Furthermore, let's recall that the OSBF strategy requires to compute the quartiles of the training set decision values in order to assign scores to stimuli. In this scenario, we stress that, for the standard P300 Speller's paradigm, stimuli corresponding to the intensification of rows and columns are considered separately; in fact, we observed that the distribution of the decision values was different for row and column stimuli. The other paradigms we considered are based on two-levels of selection: in this case, we considered stimuli corresponding to the outer and inner level together for computing the quartiles, since we observed similar distributions of the decision values. \par

\begin{table}[t]
\centering
\resizebox{\textwidth}{!}{
\begin{tabular}{lcccccccccccl}
\hline
 Dataset&NS& \#Train & \#Test &Part.&CH& Mod.&Symb.&Stim.&Max It.&SOA&OH\\
\hline
			AMUSE \cite{schreuder2011listen}& 16 & 384 & 809  &H&61&A&30&12&15&0.175&18.25\\
            CenterSpeller \cite{treder2011gaze}&13 & 220 & 538 &H&63&V&30&12&10&0.217&8.25\\
            MVEP \cite{schaeff2012exploring}&15& 270 & 606 &H&57&V&30&12&10&0.266&11.7\\
			P300Speller \cite{arico2014influence}&10&120&60&H&8&V&36&12&8&0.250&7.25\\
			ALSP300Speller \cite{riccio2013attention} &8&120&160&ALS&16&V&36&12&10&0.250&8\\
			Akimpech \cite{ledesma2010open}& 27 & 432 & 790 & H & 10 & V & 36 & 12 & 15 & 0.188 & 4 \\\hline 
		\end{tabular}}
		\caption{Dataset parameters. The following characteristics are reported: number of subjects (NS), total number of trails in the training set, total number of trials in the test set, type of paradigm, participants (part. H = healthy, ALS = amyotrophic lateral sclerosis patient), number of channels (CH), modality (mod. A = auditory, V = visual), number of possible symbols (Symb.), total number of stimuli in the selection process (for all possible levels), the maximum number of iterations in the original setting, the SOA (stimulus onset asynchrony) and the overhead (OH pre and post-stimulus pauses).}
  	\label{table:dataset parameters}
	\end{table}
\subsection{ No stopping scenario}
As a first step, we evaluate the impact of choosing the scores by solving the problem \eqref{eq:wobj}. We compare our strategy with both the classical \textbf{DV med} approach and the SBF decision function \cite{noiES:2019} where we sum up the heuristically determined scores for all the available iterations (i.e., we use it in a no stopping fashion). We build the separating hyperplane by training a linear SVM with the package Liblinear \cite{fan2008liblinear}. We try both the L1 and L2 loss, and since there is no clear winner, we report the results obtained with both the losses. Table \ref{tab:accnostop} shows the accuracy \textemdash i.e. the percentage of correctly classified characters \textemdash obtained by the different approaches.  Findings in Table \ref{tab:accnostop} show that the OSBF outperforms the other two approaches since it reaches the highest accuracy on all the datasets. Please note that the OSBF is computationally cheap since the solution of problem \eqref{eq:wobj} is extremely fast, and does not require any cross-validation phase.  In order to further improve the accuracy, we try and build the hyperplane by solving the dual problem \eqref{eq:newdualprob}. We call this approach M-SVM. In order to solve problem \eqref{eq:newdualprob}, we apply a modification of the dual coordinate  algorithm as described in the Appendix \ref{sec:dualapp}. The results obtained by OSBF applied to the M-SVM improve only on some datasets as shown in Table \ref{tab:svsboth}, with a significant improvement on the two most difficult datasets: the one containing ALS patients and AMUSE. The intuition was that it could help only when standard SVM is not ``good enough''. In order to better understand  the contribution of the new training problem, we look at the single-subject results, dividing the participants (across all the datasets) into two classes:
\begin{description}
\item[Class 1] subjects where the standard SVM problem is better than the new M-SVM;
\item[Class 2] subjects where the standard SVM problem is worse than the new M-SVM.
\end{description}

\begin{figure}
\centering
\resizebox{0.6\linewidth}{!}{
\begin{tikzpicture}
            \draw[->] (0,0)--(5,0) node[right]{}; 
            \draw[->] (0,0)--(0,5) node[above]{}; 

           \draw (-0.5,-0.5) -- (4.75,4.75); 

           \draw (3.5,3.75) node[rotate=45,font=\tiny]
               {$w^Tx=0$};
               
            \fill[red,thick] (2.75,4.5) circle (2pt);
            \draw (2.75,4.75) node[font=\small] {$x_{trg}$}; 
            
            \fill[red,thick] (2,0.75) circle (2pt);
            \draw (2,1) node[font=\small] {$x_{i}$}; 
            
            \fill[blue,thick] (0.75,3.75) circle (2pt);
            \draw (1,3.45) node[font=\small] {$z_{i} = x_{trg}-x_i$}; 

   \end{tikzpicture}
}
\caption{Given the point $x_i$, it is possible to reinterpret constraints \ref{eq:maxconstr} as adding to the training set the points $z_i$.}\label{fig:dataaug}
\end{figure}
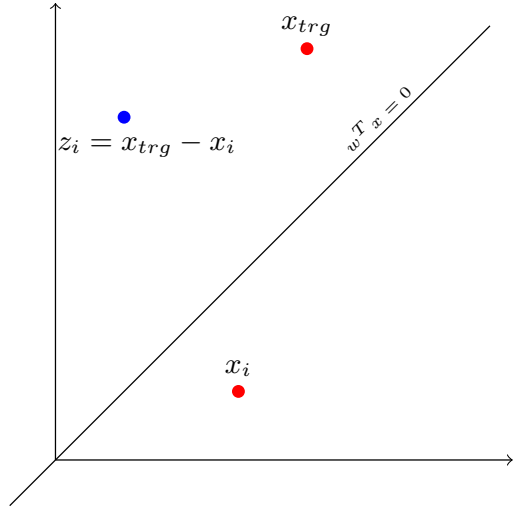

In Table \ref{tab:svsboth}, we report the average accuracy on both classes, and it is quite evident that the new training problem helps whenever the starting accuracy is not too high. When the starting accuracy is high, the performance does not change or gets worse probably for the overfitting. Interestingly, adding the constraints on the maximum decision value can be interpreted as a form of data augmentation. Indeed, if we include the bias $b$ into the vector $w$, augmenting each data point in the training set with a last component equal to $1$, we can reinterpret the constraints \eqref{eq:maxconstr} as standard sign constraints imposed on the point $z_i$. Therefore, we are augmenting our training set by adding the points $z_i$, as shown in Figure \ref{fig:dataaug}.


\begin{table}[]\centering
\resizebox{\textwidth}{!}{
\begin{tabular}{@{}lcccccc@{}}
\toprule
Dataset        & \begin{tabular}[c]{@{}c@{}}DV-med\\ L1-SVM\end{tabular} & \begin{tabular}[c]{@{}c@{}}DV-med\\L2-SVM\end{tabular}  & \begin{tabular}[c]{@{}c@{}}SBF\\L1-SVM\end{tabular} & \begin{tabular}[c]{@{}c@{}}SBF\\L2-SVM\end{tabular}     & \begin{tabular}[c]{@{}c@{}}OSBF\\L1-SVM\end{tabular}   & \begin{tabular}[c]{@{}c@{}}OSBF\\L2-SVM\end{tabular}    \\ \midrule
ALSP300Speller & 0.919          & 0.913          & 0.913      & 0.944          & 0.95           & \textbf{0.963} \\
P300Speller    & 0.95           & \textbf{0.967} & 0.95       & 0.95           & \textbf{0.967} & \textbf{0.967} \\
CenterSpeller  & \textbf{0.957} & 0.953          & 0.953      & 0.94           & 0.953          & 0.949          \\
AMUSE          & 0.741          & 0.756          & 0.752      & 0.753          & 0.763          & \textbf{0.796} \\
MVEP           & 0.754          & 0.743          & 0.739      & 0.747          & \textbf{0.777} & 0.76           \\
Akimpech       & 0.954          & 0.967          & 0.955      & \textbf{0.978} & 0.962          & \textbf{0.978} \\ \bottomrule
\end{tabular}
}\caption{Accuracy comparison between OSBF, SBF and standard approach DV-med}\label{tab:accnostop}
\end{table}

\begin{table}[]\centering
\resizebox{\textwidth}{!}{
\begin{tabular}{@{}lccc@{}}
\toprule
Dataset        & OSBF L1-SVM    & OSBF L2-SVM    & OSBF M-SVM     \\ \midrule
ALSP300Speller & 0.95           & 0.962          & \textbf{0.975} \\
P300Speller    & \textbf{0.967} & \textbf{0.967} & \textbf{0.967} \\
CenterSpeller  & 0.953          & 0.949          & \textbf{0.96}  \\
AMUSE          & 0.763          & 0.796          & \textbf{0.806} \\
MVEP           & 0.777          & 0.76           & \textbf{0.783} \\
Akimpech       & 0.962          & \textbf{0.978} & 0.974          \\ \bottomrule
\end{tabular}
}\caption{Accuracy obtained by OSBF with the different hyperplanes: standard SVM with L1 loss, standard SVM with L2 loss, and the new M-SVM obtained by solving problem \eqref{eq:newdualprob}}\label{tab:svsboth}
\end{table}

\begin{table}[]\centering
\resizebox{\textwidth}{!}{
\begin{tabular}{@{}l|cc|cc|cc@{}}
\toprule
Dataset        & \begin{tabular}[c]{@{}c@{}}Class 1\\ L2-SVM\end{tabular} & \begin{tabular}[c]{@{}c@{}}Class 1\\ M-SVM\end{tabular} & \begin{tabular}[c]{@{}c@{}}Class 2\\ L2-SVM\end{tabular} & \begin{tabular}[c]{@{}c@{}}Class 2\\ M-SVM\end{tabular} & \begin{tabular}[c]{@{}c@{}}Tot\\L2-SVM\end{tabular} & \begin{tabular}[c]{@{}c@{}}Tot\\M-SVM\end{tabular} \\ \midrule
ALSP300Speller & 0.95           & 0.9           & 0.9            & 0.975         & 0.963          & 0.975         \\
CenterSpeller  & 1              & 0.977         & 0.902          & 0.928         & 0.95           & 0.96          \\
AMUSE          & 0.865          & 0.842         & 0.695          & 0.738         & 0.796          & 0.806         \\
MVEP           & 0.859          & 0.827         & 0.714          & 0.789         & 0.76           & 0.783         \\
Akimpech       & 0.952          & 0.919         & 0.923          & 0.974         & 0.978          & 0.974         \\ \bottomrule
\end{tabular}
}\caption{Here we divide the subjects into two sets: the ones where M-SVM performs worse than the standard L2 SVM  and the ones where M-SVM performs better than the standard SVM and compute the average accuracy. Results on P300Speller dataset are not shown since in this case L2-SVM and M-SVM perform exactly the same and, so, class 1 and 2 are empty. }
\end{table}

\subsection{Early stopping scenario}

As a second step, we consider the early stopping version of both the SBF (that is the current state of the art for early stopping methods) and the OSBF.
In order to evaluate the performance of the proposed method with the respect to the number of iterations needed for an accurate classification also the theoretical Information transfer rate (ITR, bit/min) has been computed. The ITR is a communication measure based on Shannon channel theory with some simplifying assumptions. It can be computed by dividing the number of bits transmitted per trial (or bit rate, bits/trial) by the trial duration in minute. We compute the bit-rate, using the definition proposed in \cite{wolpaw1998eeg}, as: 
	\begin{equation}\label{bitrate}
	\mbox{B}=\log_{2}N + P\log_{2}P + (1-P)\log_{2}\frac{(1-P)}{(N-1)},
	\end{equation}
	where $N$ is the number of possible symbols in the speller grid and $P$ is the probability that the target symbol is accurately classified at the end of a trial. From \eqref{bitrate} the ITR is computed as:
	\begin{equation}\label{ITR}
	\mbox{ITR}= \frac{B}{\mbox{trial duration}}
	\end{equation}
	where
	\begin{equation}\label{trail duration}
	\mbox{trial duration}= \frac{\mbox{SOA} \cdot f_{s} \cdot \overline{i}}{60}\mbox{ min.}
	\end{equation} 
	In \eqref{trail duration}, SOA refers to the stimulus-onset asynchrony; $f_{s}$ represents the number of stimuli in each stimulation sequence and $\overline{i}$ is the mean number of used iterations to select a symbol. In Tables \ref{tab:accearlystop}, \ref{tab:accearlystop_both}, \ref{tab:ITRearlystop} and \ref{tab:ITRearlystop_both} the results obtained with the early stopping setting are shown. Findings in Table \ref{tab:accearlystop} further corroborates the potentials of the OSBF since its outperforms the SBF, no matter what hyperplane is used. In Table \ref{tab:accearlystop_both} we compare the early stopping results in terms of accuracy obtained with the OSBF with the different hyperplanes (L1-SVM, L2-SVM and M-SVM): we can notice that, in this case, the M-SVM reaches a higher level of accuracy than the other methods among almost all datasets. Tables \ref{tab:ITRearlystop} and \ref{tab:accearlystop_both} show the results in terms of theoretical ITR. In this case, we can see that all the strategies reach comparable results and there is not a clear winner. We can then conclude that the OSBF strategy is a more conservative approach than the SBF, since it manages to keep a high level of accuracy preserving the communication speed.

\begin{table}[]\centering
\resizebox{\textwidth}{!}{
\begin{tabular}{@{}lcccc@{}}
\toprule
Dataset        & SBF L1-SVM & SBF L2-SVM & OSBF L1-SVM    & OSBF L2-SVM    \\ \midrule
ALSP300Speller & 0.85       & 0.863      & \textbf{0.925} & \textbf{0.925} \\
P300Speller    & 0.95       & 0.95       & \textbf{0.967} & 0.95           \\
CenterSpeller  & 0.903      & 0.893      & \textbf{0.944} & 0.931          \\
AMUSE          & 0.636      & 0.647      & \textbf{0.756} & 0.744          \\
MVEP           & 0.712      & 0.712       & \textbf{0.748} & 0.744          \\
Akimpech       & 0.911      & 0.934      & 0.943          & \textbf{0.948} \\ \bottomrule
\end{tabular}
}\caption{Accuracy obtained in the early stopping setting by the SBF and OSBF using the separation hyperplane given by a linear SVM with both L1 and L2 losses}\label{tab:accearlystop}
\end{table}

\begin{table}[]\centering
\resizebox{\textwidth}{!}{
\begin{tabular}{@{}lccc@{}}
\toprule
Dataset        & OSBF L1-SVM    & OSBF L2-SVM    & OSBF M-SVM     \\ \midrule
ALSP300Speller & 0.925          & 0.925          & \textbf{0.944} \\
P300Speller    & \textbf{0.967} & 0.95           & \textbf{0.967} \\
CenterSpeller  & 0.944          & 0.931          & \textbf{0.95}  \\
AMUSE          & 0.756          & 0.744          & \textbf{0.76}  \\
MVEP           & 0.748          & 0.744          & \textbf{0.760} \\
Akimpech       & 0.943          & \textbf{0.948} & 0.942          \\ \bottomrule
\end{tabular}
}\caption{Accuracy obtained in the early stopping setting by the OSBF with the different hyperplanes: standard SVM with L1 loss, standard SVM with L2 loss, and the new M-SVM obtained by solving problem \eqref{eq:newdualprob}.}\label{tab:accearlystop_both}
\end{table}


\begin{table}[]\centering
\resizebox{\textwidth}{!}{
\begin{tabular}{@{}lcccc@{}}
\toprule
Dataset        & SBF L1-SVM      & SBF L2-SVM      & OSBF L1-SVM & OSBF L2-SVM     \\ \midrule
ALSP300Speller & 20.187          & 19.234          & 20.187      & \textbf{20.716} \\
P300Speller    & \textbf{34.38}  & 34.086          & 32.343      & 30.356          \\
CenterSpeller  & \textbf{28.372} & 28.042          & 27.76       & 27.649          \\
AMUSE          & 12.67           & 12.929          & 14.490      & \textbf{15.232} \\
MVEP           & 9.434           & 9.252           & 10.245       & \textbf{10.626}  \\
Akimpech       & 34.593          & \textbf{38.464} & 35.061      & 36.532          \\ \bottomrule
\end{tabular}
}\caption{ITR ($bit/min$) obtained in the early stopping setting by the SBF and OSBF using the separation hyperplan given by a linear SVM with both L1 and L2 losses}\label{tab:ITRearlystop}
\end{table}

\begin{table}[]\centering
\resizebox{\textwidth}{!}{
\begin{tabular}{@{}lccc@{}}
\toprule
Dataset        & OSBF L1-SVM     & OSBF L2-SVM     & OSBF M-SVM      \\ \midrule
ALSP300Speller & 20.187          & 20.716          & \textbf{21.79}  \\
P300Speller    & \textbf{32.343} & 30.356          & 32.11           \\
CenterSpeller  & 27.76           & 27.649          & \textbf{27.909} \\
AMUSE          & 14.490          & 15.232          & \textbf{15.35}  \\
MVEP           & 10.245           & \textbf{10.626}  & 10.431          \\
Akimpech       & 35.061          & \textbf{36.532} & 34.811          \\ \bottomrule
\end{tabular}
}\caption{ITR ($bit/min$) obtained in the early stopping setting by the OSBF with the different hyperplanes: standard SVM with L1 loss, standard SVM with L2 loss, and the new M-SVM obtained by solving problem \eqref{eq:newdualprob}.}\label{tab:ITRearlystop_both}
\end{table}

\section{Conclusions}\label{sec:conc}
This paper focuses on the classification problem that arises in many BCI protocols. The idea was to exploit the knowledge on the protocol in order to improve the classification accuracy, and the communication speed of the BCI. This aim has been reached by means of two different ingredients:
\begin{description}
\item[(i)] the use of a MILP problem to assign a ``reliability score'' to the classification of each stimulus in every iteration
\item[(i)] the definition of a new training problem that keeps into account that the target class is assigned to the stimuls having the maximum decision value.
\end{description}
Both novelty elements have been applied in two different scenarios: a first one where accuracy was the main focus and  all the iterations available for each subject were used both in the calibration and the online phase; a second one where the focus was to improve the communication speed, and hence an early stopping strategy was implemented in the online phase.
In order to evaluate the approaches we conducted an extensive experimentation on datasets coming from different protocols and including both healthy subjects and ALS patients. The results show how we were able to improve accuracy and ITR on all the datasets, proving once more that combining machine learning tools to problem knowledge can significantly improve performances.


\appendix
\section{Dual Coordinate Descent Algorithm}\label{sec:dualapp}
In this section we describe how we modified the Dual Coordinate Descent Algorithm proposed in \cite{dualcoordinate} in order to  find the separating hyperplane for problem~\ref{fo:maxSVM}-\ref{vinc:nonneg_max}. The Dual Coordinate Descent Algorithm basically solves the dual problem applying a Gauss Seidel decomposition method where each variable constitutes a block, and the subproblem with respect to a single variable is globally solved analytically. We adapt the algorithm by modifying the following points:
\begin{itemize}
    \item how the gradient of the objective function is computed;
    \item how the hyperplane is updated.
\end{itemize}

In particular, we can write the objective function $f$ and the separating hyperplane $w$ as:
\begin{align}\label{eq:dcd_f}
    &f = \frac{1}{2}\sum_{i=1}^{l_1}\sum_{j=1}^{l_1}\lambda_i\lambda_j y_i y_j x_i^T x_j+\sum_{i=1}^{l_1}\sum_{j=1}^{l_2}\lambda_i\rho_j y_i x_i^T z_j + \frac{1}{2}\sum_{i=1}^{l_2}\sum_{j=1}^{l_2}\rho_i\rho_j z_i z_j-\sum_{i=1}^{l_1}\lambda_i-\sum_{i=1}^{l_2}\rho_j\\\label{eq:dcd_w}
    &w = \sum_{i=1}^{l_1} \lambda_i y_i x_i + \sum_{i=1}^{l_2} \rho_i z_i
\end{align}

Let's define the vector $\alpha = \begin{bmatrix}\lambda^T & \rho^T\end{bmatrix}\in {\rm I\!R}^{l_1\times l_2}$. Equations~\ref{eq:dcd_f} and~\ref{eq:dcd_w} can equivalently be defined with respect to vector $\alpha$.
We can then express the $i$-th component of the gradient of $f(\alpha)$ as:
\begin{equation}
\nabla_i f = \begin{cases}
\sum_{j=1}^{l_1} \alpha_j y_i y_j x_i^T x_j + \sum_{j=l_1}^{l_1+l_2} \alpha_j y_i x_i^T z_j -1 \qquad \text{if }i<l_1\\
\sum_{j=1}^{l_1} \alpha_j y_j x_j^T z_j + \sum_{j=l_11}^{l_1+l_2} \alpha_j z_i^T z_j -1 \qquad \text{otherwise}
\end{cases}
\end{equation}
which can be rewritten as:
\begin{equation}
\nabla_i f = \begin{cases}
y_iw^T x_i-1 \qquad \text{if }i<l_1\\
w^T z_i -1\qquad \text{otherwise}
\end{cases}
\end{equation}

\begin{algorithm}\caption{A Dual Coordinate Descent Algorithm for problem~\ref{fo:maxSVM}-\ref{vinc:nonneg_max}}
\begin{algorithmic}
\STATE $k \leftarrow 0$
\STATE $\alpha \leftarrow \begin{bmatrix}\lambda^T & \rho^T\end{bmatrix}$
\STATE $\alpha^0 \leftarrow 0, w^0 \leftarrow 0$
\WHILE{$\alpha^k$ not optimal}
\STATE $\alpha^{k,1}\leftarrow \alpha^k, w^{k,1}\leftarrow w^k$
\FOR{$i=0$ to $l_1+l_2$}
\STATE $\nabla_i f(\alpha^{k,i}) = \begin{cases} y_i(w^{k,i})^Tx_i-1 \qquad \text{if }i<l_1;\\
(w^{k,i})^Tz_i-1 \qquad \text{otherwise}
\end{cases}$
\STATE $\nabla_i^P f(\alpha^{k,i}) = \begin{cases}
min(\nabla_i f(\alpha^{k,i}), 0) \qquad \text{if }\alpha^{k,i} = 0\\
max(\nabla_i f(\alpha^{k,i}), 0) \qquad \text{if } \alpha^{k,i} = C\\
\nabla_i f(\alpha^{k,i}) \qquad \text{otherwise}
\end{cases}$
\IF{$\nabla_i^P f(\alpha^{k,i})==0$}
\STATE $\alpha^{k+1}_{i}=\alpha^{k}_{i}, w^{k+1,i}=w^{k,i}$
\ELSE
\STATE $\alpha^{k+1}_{i} = min\left(C, max\left(0,\alpha^{k}_{i}-\frac{\nabla_i f(\alpha^{k,i})}{Q_{i,i}}\right)\right)$
\STATE $w^{k,i+1} = \begin{cases}
w^{k,i}+y_i\left(\alpha^{k+1}_{i}-\alpha^{k}_{i}\right)x_i\qquad\text{if }i<l_1;\\
w_{k,i}+\left(\alpha^{k+1}_{i}-\alpha^{k}_{i}\right)z_i\qquad\text{otherwise};
\end{cases}$
\STATE $\alpha^{k,i+1} = \begin{pmatrix}\alpha^k_i, &\dots& \alpha^{k+1}_i, & \alpha^{k}_{i+1}, &\dots & \alpha^k_{l_1+l_2}\end{pmatrix}$
\ENDIF
\ENDFOR
\ENDWHILE
\end{algorithmic}
\end{algorithm}

\section{Detailed numerical results}
As a supplement, we provide the detailed results obtained for all subjects for all considered datasets. In Table \ref{tab:det_accunostopping} we provide the results obtained in the no stopping setting, while in Tables \ref{tab:det_accuearlystopping} and \ref{tab:det_ITRearlystopping} the results for the early stopping setting are reported.

{\setlength\tabcolsep{3.7pt}\footnotesize
\begin{longtable}{@{}ll|ccccccccc@{}}
\toprule
Dataset & Subj. & \begin{tabular}[c]{@{}c@{}}DV-med\\ L1-SVM\end{tabular} & \begin{tabular}[c]{@{}c@{}}DV-med\\ L2-SVM\end{tabular} & \begin{tabular}[c]{@{}c@{}}DV-med\\ M-SVM\end{tabular} & \begin{tabular}[c]{@{}c@{}}SBF\\ L1-SVM\end{tabular} & \begin{tabular}[c]{@{}c@{}}SBF\\ L2-SVM\end{tabular} & \begin{tabular}[c]{@{}c@{}}SBF\\ M-SVM\end{tabular} & \begin{tabular}[c]{@{}c@{}}OSBF\\ L1-SVM\end{tabular} & \begin{tabular}[c]{@{}c@{}}OSBF\\ L2-SVM\end{tabular} & \begin{tabular}[c]{@{}c@{}}OSBF\\ M-SVM\end{tabular} \\ \midrule
ALSP300Speller & 1 & 0.850 & 0.900 & 0.850 & 0.900 & 0.950 & 0.900 & 0.950 & 0.900 & 0.950 \\
ALSP300Speller & 2 & 0.800 & 0.850 & 0.850 & 0.850 & 0.850 & 0.850 & 0.850 & 0.900 & 1.000 \\
ALSP300Speller & 3 & 0.800 & 0.800 & 0.850 & 0.950 & 1.000 & 0.950 & 0.900 & 0.950 & 0.900 \\
ALSP300Speller & 4 & 0.950 & 0.800 & 0.900 & 0.850 & 0.950 & 0.850 & 0.950 & 0.950 & 0.950 \\
ALSP300Speller & 5 & 1.000 & 1.000 & 1.000 & 0.950 & 0.900 & 0.950 & 1.000 & 1.000 & 1.000 \\
ALSP300Speller & 6 & 1.000 & 1.000 & 1.000 & 0.900 & 0.950 & 0.950 & 1.000 & 1.000 & 1.000 \\
ALSP300Speller & 7 & 0.950 & 0.950 & 0.950 & 0.900 & 0.950 & 0.950 & 0.950 & 1.000 & 1.000 \\
ALSP300Speller & 8 & 1.000 & 1.000 & 1.000 & 1.000 & 1.000 & 1.000 & 1.000 & 1.000 & 1.000 \\
P300Speller & 1 & 1.000 & 1.000 & 1.000 & 1.000 & 1.000 & 1.000 & 1.000 & 1.000 & 1.000 \\
P300Speller & 2 & 1.000 & 1.000 & 1.000 & 1.000 & 1.000 & 1.000 & 1.000 & 1.000 & 1.000 \\
P300Speller & 3 & 1.000 & 1.000 & 1.000 & 1.000 & 1.000 & 1.000 & 1.000 & 1.000 & 1.000 \\
P300Speller & 4 & 0.833 & 0.833 & 0.833 & 0.833 & 0.833 & 0.833 & 0.833 & 0.833 & 0.833 \\
P300Speller & 5 & 1.000 & 1.000 & 1.000 & 1.000 & 1.000 & 1.000 & 1.000 & 1.000 & 1.000 \\
P300Speller & 6 & 0.833 & 1.000 & 0.833 & 0.833 & 1.000 & 0.833 & 1.000 & 1.000 & 1.000 \\
P300Speller & 7 & 1.000 & 1.000 & 1.000 & 1.000 & 0.833 & 1.000 & 1.000 & 1.000 & 1.000 \\
P300Speller & 8 & 0.833 & 0.833 & 0.833 & 0.833 & 0.833 & 0.833 & 0.833 & 0.833 & 0.833 \\
P300Speller & 9 & 1.000 & 1.000 & 1.000 & 1.000 & 1.000 & 1.000 & 1.000 & 1.000 & 1.000 \\
P300Speller & 10 & 1.000 & 1.000 & 1.000 & 1.000 & 1.000 & 1.000 & 1.000 & 1.000 & 1.000 \\
CenterSpeller & VPiac & 0.929 & 0.952 & 0.929 & 0.905 & 0.905 & 0.905 & 0.952 & 0.929 & 0.929 \\
CenterSpeller & VPiba & 0.974 & 0.947 & 0.974 & 0.947 & 0.974 & 0.947 & 0.947 & 0.947 & 0.974 \\
CenterSpeller & VPibb & 1.000 & 1.000 & 1.000 & 0.974 & 1.000 & 0.947 & 1.000 & 1.000 & 1.000 \\
CenterSpeller & VPibc & 1.000 & 1.000 & 1.000 & 0.977 & 1.000 & 0.977 & 1.000 & 1.000 & 0.977 \\
CenterSpeller & VPibd & 0.976 & 0.927 & 0.951 & 0.976 & 0.951 & 0.976 & 0.927 & 0.927 & 0.951 \\
CenterSpeller & VPibe & 1.000 & 1.000 & 1.000 & 1.000 & 0.969 & 1.000 & 1.000 & 1.000 & 1.000 \\
CenterSpeller & VPibf & 1.000 & 1.000 & 1.000 & 0.977 & 0.977 & 0.977 & 0.977 & 1.000 & 1.000 \\
CenterSpeller & VPibg & 1.000 & 1.000 & 1.000 & 1.000 & 1.000 & 1.000 & 1.000 & 1.000 & 1.000 \\
CenterSpeller & VPibh & 0.750 & 0.769 & 0.769 & 0.827 & 0.635 & 0.788 & 0.808 & 0.788 & 0.827 \\
CenterSpeller & VPibi & 0.940 & 0.940 & 0.960 & 0.960 & 0.980 & 0.960 & 0.960 & 0.960 & 0.980 \\
CenterSpeller & VPibj & 0.930 & 0.930 & 0.907 & 0.860 & 0.884 & 0.860 & 0.837 & 0.814 & 0.837 \\
CenterSpeller & VPica & 0.973 & 0.946 & 0.973 & 1.000 & 0.973 & 1.000 & 1.000 & 1.000 & 1.000 \\
CenterSpeller & VPsaf & 0.974 & 0.974 & 1.000 & 0.974 & 0.974 & 0.949 & 0.974 & 0.974 & 1.000 \\
AMUSE & VPfar & 0.711 & 0.644 & 0.711 & 0.600 & 0.556 & 0.578 & 0.711 & 0.711 & 0.711 \\
AMUSE & VPfau & 0.845 & 0.845 & 0.810 & 0.828 & 0.862 & 0.845 & 0.862 & 0.862 & 0.879 \\
AMUSE & VPfav & 0.849 & 0.849 & 0.849 & 0.830 & 0.849 & 0.830 & 0.849 & 0.849 & 0.849 \\
AMUSE & VPfaw & 0.750 & 0.806 & 0.861 & 0.806 & 0.861 & 0.833 & 0.889 & 0.889 & 0.917 \\
AMUSE & VPfax & 0.756 & 0.756 & 0.732 & 0.768 & 0.744 & 0.720 & 0.756 & 0.756 & 0.793 \\
AMUSE & VPfaz & 0.969 & 0.969 & 0.969 & 0.969 & 0.969 & 0.969 & 0.969 & 0.969 & 0.969 \\
AMUSE & VPfca & 0.974 & 0.947 & 0.974 & 0.974 & 1.000 & 0.974 & 0.947 & 1.000 & 0.974 \\
AMUSE & VPfcb & 0.695 & 0.712 & 0.780 & 0.847 & 0.746 & 0.746 & 0.712 & 0.847 & 0.831 \\
AMUSE & VPfcc & 0.969 & 0.969 & 0.938 & 0.938 & 0.969 & 0.938 & 0.969 & 0.969 & 0.969 \\
AMUSE & VPfcd & 0.743 & 0.857 & 0.743 & 0.800 & 0.857 & 0.857 & 0.800 & 0.914 & 0.886 \\
AMUSE & VPfcg & 0.682 & 0.591 & 0.621 & 0.712 & 0.727 & 0.697 & 0.697 & 0.697 & 0.697 \\
AMUSE & VPfch & 0.259 & 0.379 & 0.397 & 0.276 & 0.241 & 0.276 & 0.241 & 0.431 & 0.500 \\
AMUSE & VPfcj & 0.391 & 0.551 & 0.536 & 0.449 & 0.377 & 0.493 & 0.536 & 0.580 & 0.638 \\
AMUSE & VPfck & 0.736 & 0.717 & 0.736 & 0.698 & 0.698 & 0.660 & 0.660 & 0.698 & 0.679 \\
AMUSE & VPfcm & 0.617 & 0.600 & 0.633 & 0.683 & 0.717 & 0.600 & 0.700 & 0.650 & 0.700 \\
AMUSE & VPkw & 0.909 & 0.909 & 0.909 & 0.848 & 0.879 & 0.909 & 0.909 & 0.909 & 0.909 \\
MVEP & VPfat & 0.821 & 0.923 & 0.897 & 0.872 & 0.923 & 0.821 & 0.949 & 0.897 & 0.897 \\
MVEP & VPgdf & 0.647 & 0.618 & 0.618 & 0.647 & 0.647 & 0.588 & 0.706 & 0.618 & 0.647 \\
MVEP & VPgdg & 0.821 & 0.744 & 0.821 & 0.692 & 0.769 & 0.718 & 0.769 & 0.744 & 0.795 \\
MVEP & VPiac & 0.667 & 0.667 & 0.727 & 0.606 & 0.697 & 0.667 & 0.727 & 0.697 & 0.697 \\
MVEP & VPiba & 0.519 & 0.500 & 0.519 & 0.537 & 0.537 & 0.593 & 0.556 & 0.500 & 0.648 \\
MVEP & VPibe & 1.000 & 0.973 & 1.000 & 0.973 & 0.973 & 0.946 & 1.000 & 1.000 & 1.000 \\
MVEP & VPibs & 0.760 & 0.780 & 0.760 & 0.820 & 0.660 & 0.740 & 0.840 & 0.840 & 0.820 \\
MVEP & VPibt & 0.810 & 0.833 & 0.810 & 0.786 & 0.833 & 0.786 & 0.762 & 0.786 & 0.786 \\
MVEP & VPibu & 0.468 & 0.553 & 0.489 & 0.426 & 0.468 & 0.447 & 0.532 & 0.489 & 0.489 \\
MVEP & VPibv & 0.833 & 0.806 & 0.833 & 0.833 & 0.750 & 0.889 & 0.889 & 0.806 & 0.944 \\
MVEP & VPibw & 0.975 & 0.925 & 0.975 & 0.975 & 0.975 & 1.000 & 0.925 & 0.975 & 0.950 \\
MVEP & VPibx & 0.917 & 0.861 & 0.917 & 0.917 & 0.917 & 0.944 & 0.972 & 0.917 & 0.944 \\
MVEP & VPiby & 0.684 & 0.632 & 0.684 & 0.737 & 0.711 & 0.684 & 0.737 & 0.763 & 0.711 \\
MVEP & VPice & 0.659 & 0.705 & 0.682 & 0.614 & 0.614 & 0.591 & 0.591 & 0.659 & 0.659 \\
MVEP & VPicv & 0.730 & 0.622 & 0.676 & 0.649 & 0.730 & 0.649 & 0.703 & 0.703 & 0.757 \\
Akimpech & ACS & 0.923 & 0.962 & 0.962 & 0.923 & 0.962 & 0.962 & 0.962 & 0.962 & 0.962 \\
Akimpech & APM & 1.000 & 1.000 & 1.000 & 1.000 & 1.000 & 1.000 & 1.000 & 1.000 & 1.000 \\
Akimpech & ASG & 1.000 & 1.000 & 1.000 & 1.000 & 1.000 & 1.000 & 1.000 & 1.000 & 1.000 \\
Akimpech & ASR & 0.973 & 0.919 & 0.865 & 0.919 & 0.973 & 0.811 & 0.946 & 0.919 & 0.865 \\
Akimpech & CLL & 0.974 & 1.000 & 1.000 & 0.949 & 0.949 & 0.949 & 0.974 & 0.923 & 0.974 \\
Akimpech & CLR & 1.000 & 1.000 & 1.000 & 1.000 & 1.000 & 1.000 & 1.000 & 1.000 & 1.000 \\
Akimpech & DCM & 0.980 & 0.980 & 0.980 & 1.000 & 0.959 & 0.959 & 1.000 & 1.000 & 0.980 \\
Akimpech & DLP & 0.957 & 0.957 & 0.913 & 0.913 & 1.000 & 1.000 & 0.957 & 1.000 & 1.000 \\
Akimpech & DMA & 0.833 & 0.800 & 0.800 & 0.833 & 0.833 & 0.833 & 0.900 & 0.867 & 0.833 \\
Akimpech & ELC & 1.000 & 1.000 & 1.000 & 1.000 & 1.000 & 1.000 & 1.000 & 1.000 & 1.000 \\
Akimpech & FSZ & 0.867 & 0.933 & 0.900 & 0.967 & 1.000 & 0.933 & 0.933 & 0.967 & 0.967 \\
Akimpech & GCE & 0.964 & 0.964 & 0.929 & 0.929 & 0.929 & 0.929 & 0.893 & 0.964 & 0.964 \\
Akimpech & ICE & 1.000 & 1.000 & 1.000 & 1.000 & 1.000 & 1.000 & 1.000 & 1.000 & 1.000 \\
Akimpech & IZH & 0.950 & 0.950 & 0.950 & 0.975 & 0.975 & 0.975 & 1.000 & 0.975 & 0.975 \\
Akimpech & JCR & 0.667 & 0.944 & 0.944 & 0.889 & 1.000 & 1.000 & 0.778 & 1.000 & 1.000 \\
Akimpech & JLD & 1.000 & 1.000 & 1.000 & 1.000 & 1.000 & 0.957 & 1.000 & 1.000 & 1.000 \\
Akimpech & JMR & 1.000 & 1.000 & 1.000 & 1.000 & 1.000 & 1.000 & 1.000 & 1.000 & 1.000 \\
Akimpech & JSC & 0.962 & 0.962 & 0.962 & 0.808 & 0.962 & 0.923 & 0.923 & 1.000 & 0.962 \\
Akimpech & JST & 1.000 & 1.000 & 1.000 & 0.971 & 1.000 & 0.971 & 0.971 & 1.000 & 1.000 \\
Akimpech & LAC & 1.000 & 1.000 & 1.000 & 1.000 & 1.000 & 1.000 & 1.000 & 1.000 & 1.000 \\
Akimpech & LAG & 0.977 & 0.953 & 0.930 & 0.977 & 0.977 & 0.977 & 0.953 & 0.977 & 0.953 \\
Akimpech & LGP & 1.000 & 1.000 & 1.000 & 1.000 & 1.000 & 1.000 & 1.000 & 1.000 & 1.000 \\
Akimpech & LPS & 1.000 & 1.000 & 1.000 & 1.000 & 1.000 & 1.000 & 1.000 & 1.000 & 1.000 \\
Akimpech & MoMR & 1.000 & 1.000 & 1.000 & 1.000 & 1.000 & 1.000 & 1.000 & 1.000 & 1.000 \\
Akimpech & PGA & 0.818 & 0.841 & 0.818 & 0.841 & 0.886 & 0.864 & 0.864 & 0.886 & 0.886 \\
Akimpech & WFG & 0.907 & 0.953 & 0.953 & 0.930 & 1.000 & 0.977 & 0.977 & 0.977 & 0.977 \\
Akimpech & XCL & 1.000 & 1.000 & 1.000 & 0.952 & 1.000 & 1.000 & 0.952 & 1.000 & 1.000 \\ \bottomrule
\caption{Detail of the Accuracy results obtains with all no stopping framework mentioned}\label{tab:det_accunostopping}
\end{longtable}}

{\setlength\tabcolsep{3.7pt}\small
\begin{longtable}{@{}ll|ccccccccc@{}}
\toprule
Dataset & Subj. & \begin{tabular}[c]{@{}c@{}}SBF\\ L1-SVM\end{tabular} & \begin{tabular}[c]{@{}l@{}}SBF\\ L2-SVM\end{tabular} & \begin{tabular}[c]{@{}c@{}}SBF\\ M-SVM\end{tabular} & \begin{tabular}[c]{@{}c@{}}OSBF\\ L1-SVM\end{tabular} & \begin{tabular}[c]{@{}l@{}}OSBF\\ L2-SVM\end{tabular} & \begin{tabular}[c]{@{}c@{}}OSBF\\ M-SVM\end{tabular} \\ \midrule
ALSP300Speller & 1 & 0.950 & 0.800 & 0.900 & 0.950 & 0.850 & 1.000 \\
ALSP300Speller & 2 & 0.750 & 0.750 & 0.700 & 0.800 & 0.900 & 0.850 \\
ALSP300Speller & 3 & 0.850 & 0.850 & 0.750 & 0.800 & 0.950 & 0.800 \\
ALSP300Speller & 4 & 0.750 & 0.900 & 0.750 & 0.900 & 0.900 & 0.900 \\
ALSP300Speller & 5 & 0.850 & 0.850 & 0.850 & 1.000 & 0.950 & 1.000 \\
ALSP300Speller & 6 & 0.800 & 0.850 & 0.850 & 1.000 & 0.900 & 1.000 \\
ALSP300Speller & 7 & 0.900 & 0.900 & 0.900 & 0.950 & 0.950 & 1.000 \\
ALSP300Speller & 8 & 0.950 & 1.000 & 1.000 & 1.000 & 1.000 & 1.000 \\
P300Speller & 1 & 0.833 & 0.833 & 0.833 & 1.000 & 1.000 & 1.000 \\
P300Speller & 2 & 1.000 & 1.000 & 1.000 & 1.000 & 1.000 & 1.000 \\
P300Speller & 3 & 1.000 & 1.000 & 1.000 & 1.000 & 0.833 & 1.000 \\
P300Speller & 4 & 0.833 & 0.833 & 0.833 & 0.833 & 0.833 & 0.833 \\
P300Speller & 5 & 1.000 & 1.000 & 1.000 & 1.000 & 1.000 & 1.000 \\
P300Speller & 6 & 0.833 & 1.000 & 0.833 & 1.000 & 1.000 & 1.000 \\
P300Speller & 7 & 1.000 & 0.833 & 1.000 & 1.000 & 1.000 & 1.000 \\
P300Speller & 8 & 1.000 & 1.000 & 1.000 & 0.833 & 0.833 & 0.833 \\
P300Speller & 9 & 1.000 & 1.000 & 1.000 & 1.000 & 1.000 & 1.000 \\
P300Speller & 10 & 1.000 & 1.000 & 1.000 & 1.000 & 1.000 & 1.000 \\
CenterSpeller & VPiac & 0.833 & 0.857 & 0.810 & 0.881 & 0.881 & 0.857 \\
CenterSpeller & VPiba & 0.921 & 0.921 & 0.921 & 0.947 & 0.947 & 0.974 \\
CenterSpeller & VPibb & 0.947 & 0.921 & 0.947 & 1.000 & 0.974 & 1.000 \\
CenterSpeller & VPibc & 0.932 & 0.909 & 0.909 & 1.000 & 1.000 & 0.977 \\
CenterSpeller & VPibd & 0.878 & 0.927 & 0.878 & 0.927 & 0.951 & 0.927 \\
CenterSpeller & VPibe & 0.969 & 0.938 & 0.969 & 0.969 & 1.000 & 0.969 \\
CenterSpeller & VPibf & 0.907 & 0.907 & 0.907 & 1.000 & 0.930 & 1.000 \\
CenterSpeller & VPibg & 1.000 & 1.000 & 1.000 & 1.000 & 1.000 & 1.000 \\
CenterSpeller & VPibh & 0.769 & 0.615 & 0.788 & 0.827 & 0.788 & 0.846 \\
CenterSpeller & VPibi & 0.900 & 0.940 & 0.960 & 0.960 & 0.940 & 0.960 \\
CenterSpeller & VPibj & 0.814 & 0.837 & 0.837 & 0.837 & 0.791 & 0.837 \\
CenterSpeller & VPica & 0.919 & 0.892 & 0.946 & 0.973 & 0.946 & 1.000 \\
CenterSpeller & VPsaf & 0.949 & 0.949 & 0.949 & 0.949 & 0.949 & 1.000 \\
AMUSE & VPfar & 0.489 & 0.489 & 0.489 & 0.689 & 0.644 & 0.644 \\
AMUSE & VPfau & 0.759 & 0.828 & 0.741 & 0.828 & 0.862 & 0.845 \\
AMUSE & VPfav & 0.717 & 0.830 & 0.811 & 0.830 & 0.830 & 0.811 \\
AMUSE & VPfaw & 0.528 & 0.667 & 0.694 & 0.806 & 0.833 & 0.861 \\
AMUSE & VPfax & 0.646 & 0.500 & 0.634 & 0.732 & 0.683 & 0.683 \\
AMUSE & VPfaz & 0.938 & 0.875 & 0.938 & 0.969 & 0.938 & 0.969 \\
AMUSE & VPfca & 0.737 & 0.763 & 0.658 & 0.974 & 0.974 & 0.974 \\
AMUSE & VPfcb & 0.559 & 0.576 & 0.644 & 0.729 & 0.712 & 0.814 \\
AMUSE & VPfcc & 0.844 & 0.844 & 0.906 & 0.969 & 0.938 & 0.875 \\
AMUSE & VPfcd & 0.714 & 0.714 & 0.657 & 0.829 & 0.857 & 0.829 \\
AMUSE & VPfcg & 0.515 & 0.606 & 0.470 & 0.667 & 0.636 & 0.636 \\
AMUSE & VPfch & 0.293 & 0.224 & 0.293 & 0.362 & 0.362 & 0.414 \\
AMUSE & VPfcj & 0.435 & 0.362 & 0.493 & 0.449 & 0.478 & 0.565 \\
AMUSE & VPfck & 0.547 & 0.528 & 0.528 & 0.642 & 0.679 & 0.679 \\
AMUSE & VPfcm & 0.583 & 0.667 & 0.417 & 0.683 & 0.600 & 0.617 \\
AMUSE & VPkw & 0.879 & 0.879 & 0.909 & 0.939 & 0.879 & 0.939 \\
MVEP & VPfat & 0.846 & 0.974 & 0.769 & 0.923 & 0.897 & 0.897 \\
MVEP & VPgdf & 0.529 & 0.588 & 0.441 & 0.588 & 0.618 & 0.618 \\
MVEP & VPgdg & 0.564 & 0.769 & 0.667 & 0.795 & 0.744 & 0.795 \\
MVEP & VPiac & 0.727 & 0.667 & 0.545 & 0.606 & 0.667 & 0.667 \\
MVEP & VPiba & 0.556 & 0.519 & 0.611 & 0.593 & 0.463 & 0.593 \\
MVEP & VPibe & 0.892 & 0.919 & 0.892 & 1.000 & 1.000 & 1.000 \\
MVEP & VPibs & 0.800 & 0.580 & 0.600 & 0.680 & 0.800 & 0.780 \\
MVEP & VPibt & 0.786 & 0.833 & 0.786 & 0.786 & 0.786 & 0.786 \\
MVEP & VPibu & 0.404 & 0.468 & 0.447 & 0.511 & 0.468 & 0.468 \\
MVEP & VPibv & 0.889 & 0.750 & 0.889 & 0.944 & 0.861 & 0.917 \\
MVEP & VPibw & 0.975 & 0.950 & 0.875 & 1.000 & 1.000 & 0.950 \\
MVEP & VPibx & 0.833 & 0.833 & 0.806 & 0.917 & 0.917 & 0.889 \\
MVEP & VPiby & 0.711 & 0.684 & 0.658 & 0.605 & 0.684 & 0.711 \\
MVEP & VPice & 0.545 & 0.523 & 0.523 & 0.591 & 0.614 & 0.659 \\
MVEP & VPicv & 0.622 & 0.622 & 0.568 & 0.676 & 0.649 & 0.676 \\
Akimpech & ACS & 0.808 & 0.923 & 0.846 & 0.846 & 0.846 & 0.885 \\
Akimpech & APM & 1.000 & 1.000 & 1.000 & 1.000 & 1.000 & 1.000 \\
Akimpech & ASG & 1.000 & 0.962 & 0.962 & 0.923 & 0.962 & 0.962 \\
Akimpech & ASR & 0.811 & 0.946 & 0.649 & 0.946 & 0.892 & 0.838 \\
Akimpech & CLL & 0.846 & 0.872 & 0.923 & 0.872 & 0.923 & 0.974 \\
Akimpech & CLR & 1.000 & 1.000 & 1.000 & 1.000 & 1.000 & 1.000 \\
Akimpech & DCM & 0.939 & 0.898 & 0.918 & 0.959 & 0.959 & 0.939 \\
Akimpech & DLP & 0.826 & 0.870 & 1.000 & 0.957 & 0.957 & 0.957 \\
Akimpech & DMA & 0.900 & 0.800 & 0.800 & 0.900 & 0.867 & 0.867 \\
Akimpech & ELC & 1.000 & 1.000 & 1.000 & 1.000 & 1.000 & 1.000 \\
Akimpech & FSZ & 0.967 & 1.000 & 0.900 & 0.933 & 0.967 & 0.967 \\
Akimpech & GCE & 0.857 & 0.893 & 0.821 & 0.929 & 0.893 & 0.893 \\
Akimpech & ICE & 0.917 & 1.000 & 0.875 & 1.000 & 1.000 & 1.000 \\
Akimpech & IZH & 0.875 & 0.850 & 0.800 & 0.925 & 0.925 & 0.900 \\
Akimpech & JCR & 0.889 & 0.944 & 0.833 & 1.000 & 0.944 & 0.889 \\
Akimpech & JLD & 1.000 & 1.000 & 0.913 & 1.000 & 1.000 & 1.000 \\
Akimpech & JMR & 0.923 & 0.923 & 0.923 & 1.000 & 1.000 & 0.962 \\
Akimpech & JSC & 0.654 & 0.846 & 0.885 & 0.808 & 0.923 & 0.885 \\
Akimpech & JST & 1.000 & 1.000 & 1.000 & 1.000 & 1.000 & 1.000 \\
Akimpech & LAC & 1.000 & 1.000 & 1.000 & 1.000 & 1.000 & 1.000 \\
Akimpech & LAG & 0.930 & 0.907 & 0.930 & 0.907 & 0.953 & 0.907 \\
Akimpech & LGP & 1.000 & 1.000 & 1.000 & 1.000 & 1.000 & 1.000 \\
Akimpech & LPS & 0.800 & 0.800 & 0.800 & 0.800 & 0.800 & 0.800 \\
Akimpech & MoMR & 1.000 & 0.941 & 0.941 & 1.000 & 0.941 & 1.000 \\
Akimpech & PGA & 0.818 & 0.864 & 0.841 & 0.841 & 0.886 & 0.886 \\
Akimpech & WFG & 0.930 & 0.977 & 0.907 & 0.977 & 0.953 & 0.977 \\
Akimpech & XCL & 0.905 & 1.000 & 0.952 & 0.952 & 1.000 & 0.952 \\ \bottomrule
\caption{Detail of the Accuracy results obtains with all early stopping framework mentioned}\label{tab:det_accuearlystopping}
\end{longtable}}

{\setlength\tabcolsep{3.7pt}\small
\begin{longtable}{@{}ll|ccccccccc@{}}
\toprule
Dataset & Subj. & \begin{tabular}[c]{@{}l@{}}SBF\\ L1-SVM\end{tabular} & \begin{tabular}[c]{@{}c@{}}SBF\\ L2-SVM\end{tabular} & \begin{tabular}[c]{@{}c@{}}SBF\\ M-SVM\end{tabular} & \begin{tabular}[c]{@{}l@{}}OSBF\\ L1-SVM\end{tabular} & \begin{tabular}[c]{@{}c@{}}OSBF\\ L2-SVM\end{tabular} & \begin{tabular}[c]{@{}c@{}}OSBF\\ M-SVM\end{tabular} \\ \midrule
ALSP300Speller & 1 & 16.599 & 15.734 & 19.593 & 16.599 & 15.239 & 20.175 \\
ALSP300Speller & 2 & 12.222 & 12.556 & 11.111 & 12.222 & 13.845 & 13.722 \\
ALSP300Speller & 3 & 16.492 & 19.073 & 16.298 & 16.492 & 21.032 & 15.210 \\
ALSP300Speller & 4 & 16.424 & 12.501 & 13.826 & 16.424 & 15.879 & 16.836 \\
ALSP300Speller & 5 & 20.577 & 20.771 & 21.207 & 20.577 & 20.565 & 22.236 \\
ALSP300Speller & 6 & 21.654 & 19.193 & 18.605 & 21.654 & 20.305 & 22.356 \\
ALSP300Speller & 7 & 18.508 & 21.477 & 23.429 & 18.508 & 20.565 & 23.236 \\
ALSP300Speller & 8 & 39.018 & 32.566 & 36.928 & 39.018 & 38.296 & 40.548 \\
P300Speller & 1 & 25.132 & 23.148 & 25.132 & 31.020 & 30.263 & 31.815 \\
P300Speller & 2 & 38.774 & 40.025 & 40.025 & 41.359 & 40.025 & 41.359 \\
P300Speller & 3 & 32.652 & 22.977 & 30.263 & 31.815 & 19.547 & 29.542 \\
P300Speller & 4 & 33.831 & 36.650 & 36.650 & 24.433 & 25.132 & 24.433 \\
P300Speller & 5 & 41.359 & 44.314 & 41.359 & 37.599 & 37.599 & 37.599 \\
P300Speller & 6 & 25.132 & 31.020 & 26.655 & 27.573 & 29.542 & 29.542 \\
P300Speller & 7 & 27.573 & 21.454 & 27.573 & 28.855 & 26.974 & 27.573 \\
P300Speller & 8 & 30.263 & 28.855 & 31.020 & 21.990 & 19.991 & 20.456 \\
P300Speller & 9 & 41.359 & 42.786 & 41.359 & 34.466 & 34.466 & 34.466 \\
P300Speller & 10 & 47.722 & 49.631 & 47.722 & 44.314 & 40.025 & 44.314 \\
CenterSpeller & VPiac & 20.249 & 26.589 & 23.718 & 23.546 & 24.404 & 21.732 \\
CenterSpeller & VPiba & 26.307 & 26.794 & 26.794 & 27.868 & 26.982 & 30.014 \\
CenterSpeller & VPibb & 31.684 & 31.049 & 31.816 & 33.617 & 30.699 & 32.973 \\
CenterSpeller & VPibc & 29.426 & 27.713 & 28.289 & 30.473 & 31.336 & 29.083 \\
CenterSpeller & VPibd & 23.146 & 27.908 & 24.662 & 23.647 & 26.704 & 23.576 \\
CenterSpeller & VPibe & 36.558 & 33.161 & 37.370 & 32.653 & 33.707 & 32.973 \\
CenterSpeller & VPibf & 32.471 & 26.430 & 31.821 & 32.140 & 30.004 & 31.824 \\
CenterSpeller & VPibg & 43.509 & 43.085 & 44.608 & 39.607 & 43.725 & 38.569 \\
CenterSpeller & VPibh & 18.407 & 11.128 & 16.535 & 16.797 & 15.199 & 19.050 \\
CenterSpeller & VPibi & 28.770 & 31.514 & 31.356 & 27.147 & 27.185 & 26.935 \\
CenterSpeller & VPibj & 22.652 & 24.898 & 24.808 & 18.391 & 16.365 & 18.844 \\
CenterSpeller & VPica & 25.320 & 24.957 & 27.354 & 25.990 & 24.959 & 25.704 \\
CenterSpeller & VPsaf & 30.342 & 29.324 & 31.062 & 28.999 & 28.167 & 31.544 \\
AMUSE & VPfar & 5.549 & 5.370 & 6.552 & 8.610 & 9.986 & 8.314 \\
AMUSE & VPfau & 14.682 & 16.384 & 17.185 & 14.834 & 16.699 & 16.139 \\
AMUSE & VPfav & 15.535 & 14.670 & 15.456 & 19.101 & 21.385 & 19.395 \\
AMUSE & VPfaw & 8.541 & 12.963 & 12.811 & 14.112 & 16.804 & 16.256 \\
AMUSE & VPfax & 10.425 & 8.435 & 10.154 & 11.186 & 10.287 & 9.579 \\
AMUSE & VPfaz & 29.548 & 31.219 & 31.078 & 33.581 & 30.235 & 37.179 \\
AMUSE & VPfca & 19.655 & 20.983 & 17.787 & 27.462 & 27.767 & 27.462 \\
AMUSE & VPfcb & 7.303 & 8.242 & 10.630 & 8.356 & 10.413 & 11.195 \\
AMUSE & VPfcc & 20.129 & 22.135 & 24.316 & 22.509 & 25.576 & 22.971 \\
AMUSE & VPfcd & 15.004 & 14.126 & 13.374 & 13.763 & 15.183 & 14.970 \\
AMUSE & VPfcg & 8.621 & 10.283 & 7.204 & 9.691 & 9.228 & 9.166 \\
AMUSE & VPfch & 1.838 & 1.146 & 1.469 & 2.819 & 2.862 & 3.136 \\
AMUSE & VPfcj & 3.197 & 2.200 & 3.654 & 4.565 & 5.378 & 7.027 \\
AMUSE & VPfck & 9.682 & 9.056 & 9.368 & 9.594 & 10.557 & 10.408 \\
AMUSE & VPfcm & 5.763 & 7.058 & 4.565 & 9.227 & 7.871 & 8.002 \\
AMUSE & VPkw & 27.244 & 22.588 & 28.013 & 22.433 & 23.480 & 24.398 \\
MVEP & VPfat & 14.958 & 14.291 & 13.818 & 17.155 & 15.372 & 14.972 \\
MVEP & VPgdf & 5.504 & 5.419 & 3.910 & 5.314 & 5.845 & 6.016 \\
MVEP & VPgdg & 6.586 & 11.509 & 9.982 & 8.776 & 8.822 & 9.048 \\
MVEP & VPiac & 11.131 & 8.422 & 7.097 & 6.886 & 8.350 & 8.398 \\
MVEP & VPiba & 4.186 & 5.435 & 7.175 & 5.528 & 4.035 & 6.015 \\
MVEP & VPibe & 16.537 & 17.687 & 16.239 & 18.598 & 21.599 & 18.751 \\
MVEP & VPibs & 8.185 & 6.961 & 6.306 & 6.972 & 9.618 & 9.155 \\
MVEP & VPibt & 9.232 & 10.368 & 9.426 & 12.149 & 11.743 & 12.149 \\
MVEP & VPibu & 2.253 & 3.136 & 2.854 & 4.108 & 3.430 & 3.672 \\
MVEP & VPibv & 10.438 & 7.861 & 10.650 & 14.113 & 12.383 & 13.729 \\
MVEP & VPibw & 13.403 & 13.192 & 17.706 & 17.485 & 19.418 & 15.786 \\
MVEP & VPibx & 15.977 & 14.671 & 15.071 & 16.467 & 16.972 & 15.894 \\
MVEP & VPiby & 10.564 & 7.123 & 9.098 & 6.992 & 8.197 & 8.767 \\
MVEP & VPice & 5.280 & 4.897 & 5.019 & 5.710 & 6.264 & 6.892 \\
MVEP & VPicv & 7.271 & 7.813 & 6.522 & 7.422 & 7.343 & 7.215 \\
Akimpech & ACS & 23.354 & 29.594 & 27.111 & 22.840 & 23.990 & 24.666 \\
Akimpech & APM & 51.115 & 53.921 & 54.563 & 50.551 & 51.691 & 48.586 \\
Akimpech & ASG & 40.801 & 43.749 & 42.566 & 36.551 & 40.254 & 41.175 \\
Akimpech & ASR & 22.004 & 20.972 & 13.311 & 22.561 & 20.642 & 16.259 \\
Akimpech & CLL & 24.106 & 29.092 & 30.928 & 23.146 & 29.745 & 32.356 \\
Akimpech & CLR & 52.884 & 47.413 & 45.833 & 45.833 & 41.666 & 39.285 \\
Akimpech & DCM & 45.059 & 42.503 & 42.313 & 42.218 & 43.689 & 40.103 \\
Akimpech & DLP & 21.249 & 31.112 & 28.765 & 28.278 & 28.668 & 25.819 \\
Akimpech & DMA & 18.975 & 25.073 & 24.142 & 25.724 & 28.174 & 26.192 \\
Akimpech & ELC & 53.088 & 50.737 & 53.710 & 50.182 & 52.083 & 48.076 \\
Akimpech & FSZ & 34.140 & 45.379 & 29.082 & 32.874 & 35.871 & 34.417 \\
Akimpech & GCE & 25.441 & 28.303 & 26.500 & 25.405 & 25.961 & 24.513 \\
Akimpech & ICE & 35.864 & 38.841 & 28.723 & 35.256 & 34.203 & 33.536 \\
Akimpech & IZH & 31.358 & 32.417 & 28.353 & 31.212 & 32.067 & 29.545 \\
Akimpech & JCR & 17.051 & 21.163 & 22.774 & 23.031 & 23.047 & 19.216 \\
Akimpech & JLD & 49.818 & 45.833 & 39.010 & 44.788 & 47.577 & 42.177 \\
Akimpech & JMR & 29.005 & 41.494 & 39.660 & 37.878 & 41.044 & 39.746 \\
Akimpech & JSC & 17.758 & 27.944 & 25.540 & 19.469 & 28.301 & 25.784 \\
Akimpech & JST & 53.710 & 57.053 & 55.220 & 51.691 & 54.347 & 51.115 \\
Akimpech & LAC & 58.760 & 57.772 & 60.043 & 55.443 & 50.182 & 52.083 \\
Akimpech & LAG & 39.163 & 38.809 & 41.354 & 36.548 & 41.729 & 36.667 \\
Akimpech & LGP & 41.540 & 50.737 & 51.497 & 46.928 & 51.497 & 50.551 \\
Akimpech & LPS & 21.166 & 23.951 & 22.754 & 22.754 & 22.754 & 24.598 \\
Akimpech & MoMR & 36.569 & 37.083 & 33.674 & 35.714 & 32.850 & 35.166 \\
Akimpech & PGA & 17.631 & 24.322 & 21.851 & 22.807 & 24.922 & 22.353 \\
Akimpech & WFG & 31.372 & 41.167 & 37.770 & 34.324 & 33.587 & 35.640 \\
Akimpech & XCL & 41.036 & 52.083 & 46.840 & 42.640 & 45.833 & 40.279 \\ \bottomrule
\caption{Detail of the ITR ($bit/min$) results obtains with all early stopping framework mentioned}\label{tab:det_ITRearlystopping}
\end{longtable}}

\end{document}